\documentclass[11pt, letterpaper]{article}

\usepackage{pbg_template}
\usepackage{parskip}
\usepackage{float}

\usepackage[utf8]{inputenc} 
\usepackage[T1]{fontenc}    
\usepackage{hyperref}       
\usepackage{url}            
\usepackage{booktabs}       
\usepackage{amsfonts}       
\usepackage{nicefrac}       
\usepackage{microtype}      
\usepackage{xcolor}         

\usepackage{wrapfig}
\usepackage{caption}
\usepackage{multirow}
\usepackage{tabularx}

\RequirePackage[authoryear,square]{natbib}
\usepackage{amsmath}
\usepackage{amssymb}
\usepackage{mathtools}      
\usepackage{amsthm}         

\usepackage{algorithm}      
\usepackage{algpseudocode}  
\usepackage{xcolor}         

\usepackage[capitalize,noabbrev]{cleveref} 

\theoremstyle{plain}
\newtheorem{theorem}{Theorem}[section]
\newtheorem{proposition}[theorem]{Proposition}

\newtheorem{corollary}[theorem]{Corollary}

\definecolor{pennred}{RGB}{153,0,0}
\definecolor{pennblue}{RGB}{102, 194, 165}  

\hypersetup{
    colorlinks=true,
    linkcolor=blue,  
    citecolor=blue,  
}

\theoremstyle{definition}

\theoremstyle{remark}

\title{Minimal-Action Discrete Schrödinger Bridge Matching for Peptide Sequence Design}

\author{Shrey Goel,\textsuperscript{1} 
    Pranam Chatterjee\textsuperscript{2,3,\dag}
    
    \vspace{1em}
    \normalfont \small
    \textsuperscript{1}Department of Computer Science, Duke University\\ 
    \textsuperscript{2}Department of Computer and Information Science, University of Pennsylvania\\ 
    \textsuperscript{3}Department of Bioengineering, University of Pennsylvania  
    
    \vspace{0.5em}
    \textbf{Correspondence:} \href{mailto:pranam@seas.upenn.edu}{\texttt{pranam@seas.upenn.edu}} \\
    \textbf{Code:} \href{https://huggingface.co/ChatterjeeLab/MadSBM}{\texttt{https://huggingface.co/ChatterjeeLab/MadSBM}}
}

\begin{document}

\maketitle

\begin{abstract}
Generative modeling of peptide sequences requires navigating a discrete and highly constrained space in which many intermediate states are chemically implausible or unstable. Existing discrete diffusion and flow-based methods rely on reversing fixed corruption processes or following prescribed probability paths, which can force generation through low-likelihood regions and require countless sampling steps. We introduce \textbf{M}inimal-\textbf{a}ction \textbf{d}iscrete \textbf{S}chrödinger \textbf{B}ridge \textbf{M}atching (\textbf{MadSBM}), a rate-based generative framework for peptide design that formulates generation as a controlled continuous-time Markov process on the amino-acid edit graph.  To yield probability trajectories that remain near high-likelihood sequence neighborhoods throughout generation, MadSBM 1) defines generation relative to a biologically informed reference process derived from pre-trained protein language model logits and 2) learns a time-dependent control field that biases transition rates to produce low-action transport paths from a masked prior to the data distribution. We finally introduce guidance to the MadSBM sampling procedure towards a specific functional objective, expanding the design space of therapeutic peptides; to our knowledge, this represents the first-ever application of discrete classifier guidance to Schr\"odinger bridge-based generative models.
\end{abstract}

\section{Introduction}
Generative modeling has become a central tool for peptide and protein design, enabling data-driven discovery of binders, modulators, and degraders across a broad range of biological targets \citep{bhat2025pepprclip,brixi2023saltnpeppr,chen2025moppitv3,chen2025pepmlm,goel2025tokenlevel,hong2025duab,pacesa2025bindcraft,stark2025boltzgen,starkdirichlet,tang2025peptune,tang2025tr2d2,vincoff2025fusonplm,vincoff2025soapia,zhang2025metalorian}. Recent models based on autoregressive decoding, diffusion, and flow matching generate sequences by reversing a fixed corruption process or by interpolating in probability space \citep{chen2025areuredi, chen2025mogdfm, tang2025branchsbm, tang2025entangledsbm, tang2025gumbelfm, zhang2025scoobdoob, starkdirichlet}. These approaches have produced strong results, yet they impose a prescribed trajectory between noise and data that often forces the generative process through low-likelihood or unstable intermediates \citep{ho2020denoising, song2020score, sahoo2024mdlm, lipman2022flow, domingo2024stochastic, domingo2024adjoint}. Biological sequences are highly sensitive to local perturbations, so hand-crafted probability paths can misalign with the manifold of functional peptides and complicate both unconditional and guided generation.

A principled alternative is to treat generation as a transport problem \citep{chen2021optimal}. Optimal transport and Schrödinger bridge formulations define stochastic paths that connect a prior distribution to the data while minimizing an action functional \citep{schrodinger1932theorie, leonard2013survey, de2021diffusion}. In continuous settings this yields smoother and more stable generative trajectories, but classical Schrödinger bridge solvers rely on forward and backward iterations that are difficult to scale to discrete sequence spaces with large vocabularies and edit-based dynamics \citep{shi2023diffsbm, vargas2021solving, genevay2018learning}. Although recent work has extended Schrödinger bridge methods to discrete domains, including iterative forward-backward projections for categorical variables and diffusion-style CTMC bridges, these approaches are validated on VQ image tokens or graph data that do not directly align with long, language-like sequence spaces \citep{kim2024discrete, ksenofontov2025categorical}. Recent simulation-free transport methods demonstrate that optimal transport paths can be approximated without fully solving the forward and backward bridge \citep{tong2023simulation}, suggesting that a discrete variant may be feasible for biological sequences.

In this work, we introduce \textbf{Minimal-Action Discrete Schrödinger Bridge Matching (MadSBM)}, a discrete generative framework that models peptide sequence evolution as a controlled continuous-time Markov chain on the amino-acid edit graph.

\begin{figure}[h]
    \centering
    \includegraphics[width=0.7\linewidth]{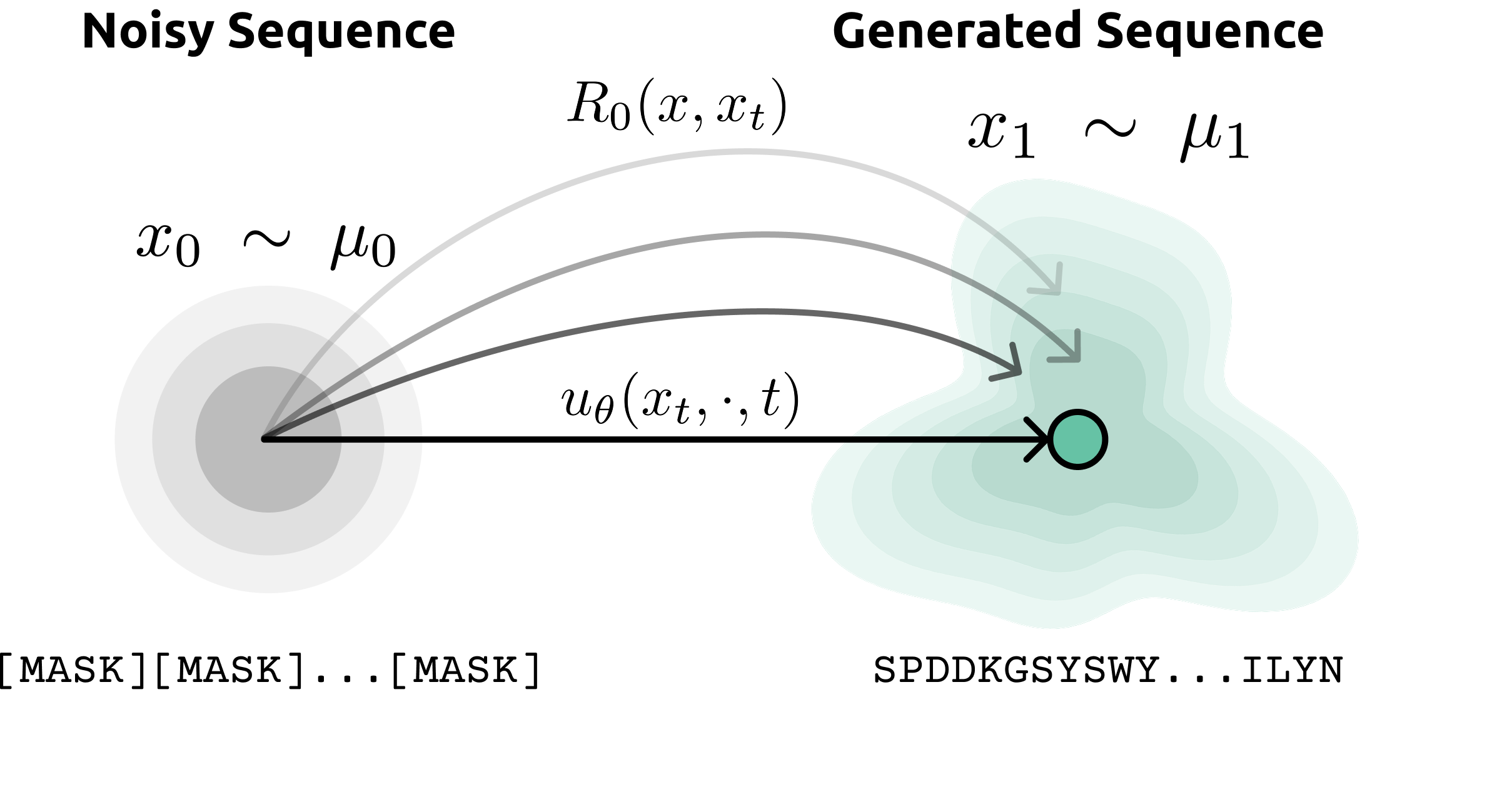}
    \caption{\textbf{Overview of MadSBM}. We leverage a principled reference process $R_0$ so the MadSBM model requires only a lightweight time-conditioned control field $u_\theta$ to steer samples toward high-likelihood regions of the sequence space.}
    \label{fig:prob_paths}
\end{figure}

MadSBM tilts a biologically informed reference process through learnable control fields so that generation follows low-action transport paths that remain near high-probability regions of sequence space, enabling efficient sampling without relying on fixed interpolation schemes as in diffusion or flow models. Our main contributions are threefold:
\begin{enumerate}
    \item  We formulate peptide sequence generation as a \textbf{discrete Schrödinger bridge}, casting design as minimal-action stochastic transport between a noisy prior and the data distribution over discrete amino-acid sequences.
    \item We introduce a \textbf{simulation-free learning rule} for Schrödinger bridge control that reduces training to a simple cross-entropy objective, avoiding explicit forward–backward bridge solvers while remaining consistent with minimal-action transport.
    \item We present \textbf{MadSBM}, a controllable rate-based generative model for biological sequences that improves sample efficiency and stability over discrete diffusion baselines and supports objective-guided peptide design under low sampling budgets.
\end{enumerate}

\section{Related Works}

A growing line of bridge-based methods have interpreted steering generative models as sampling from an exponentially tilted path measure of the form
$p_{\text{target}} \propto p_{\text{base}} \exp(r(x))$, where $p_{\text{base}}$ is a pre-trained generative process and $r(x)$ is a reward or energy function \citep{domingo2024adjoint, potaptchik2025tilt}.  
In this view, improved generative quality is achieved by modifying the transition dynamics of an \textit{existing} diffusion model to favor high-reward trajectories, and recent work has extend this idea to discrete domains by tilting a pre-trained discrete diffusion models \citep{tang2025tr2d2}. Critically, these approaches assume access to a learned base path measure that already transports noise to data and apply control as a post hoc modification of the learned process.

In contrast, MadSBM addresses \textit{de novo} generation in discrete sequence spaces, where no pre-trained path measure is available. We show that a carefully designed reference stochastic process leads to a simple training objective for learning a \textit{control field} that tilts this process into a Schrödinger bridge between the noise and data distribution.

\section{Preliminaries}
\label{sec:preliminaries}

\subsection{Discrete Sequences and Markov Dynamics}
Let $x=(x_1,\dots,x_L)\in\{0,1\}^{L\times\mathcal{V}}$ denote a discrete sequence of length $L$, where each token is represented as a one-hot vector over a vocabulary $\mathcal{V}$. The time evolution of a sequence is modeled as a Continuous-Time Markov Chain (CTMC) with generator $R$, where off-diagonal entries $R(x,x')$ define instantaneous transition rates. We use this to introduce the reference and controlled processes that shape the generator-induced dynamics toward target distributions.


\subsection{Reference and Controlled Processes}

\paragraph{Reference Dynamics}
Schrödinger bridge models begin by specifying a simple reference path measure $\mathbb{P}_0$ \citep{schrodinger1932theorie, leonard2013survey}. In discrete sequence spaces, a natural baseline is an uninformed CTMC generator $R_0$ corresponding to a uniform random walk over sequence tokens.
\begin{equation}
\label{eq:uninformed_ctmc}
R_0(x,x')=
\begin{cases}
\frac{1}{L|\mathcal{V}|} & x' \neq x,\\
-\sum_{y\neq x}R_0(x,y) & x'=x.
\end{cases}
\end{equation}
The induced process evolves tokens independently and quickly drifts into low-probability and biologically implausible regions of sequence space. Although not suitable for direct sequence generation, this process provides a canonical reference relative to which a \textit{controlled process} can be defined.

\paragraph{Exponential Tilting}
We formulate generation as steering the reference process toward the data distribution along an optimal transport path, \textit{i.e.} matching endpoint marginals while minimizing deviation from $\mathbb{P}_0$. To this end, we parameterize a controlled generator $R_u$ via an exponential tilt of the reference rates:
\begin{equation}
R_u(x,x')=R_0(x,x')\exp\big(u(x,x')\big).
\end{equation}
Here, the control field $u:\mathcal{V}^L\times\mathcal{V}^L\to\mathbb{R}$ acts as a learnable log-potential on transitions, which are combined with the reference generator to define time-dependent transition rates. This construction ensures that the controlled path measure $\mathbb{P}_u$ remains absolutely continuous with respect to $\mathbb{P}_0$, as required for the Schrödinger bridge formulation, while allowing the model to encode context-dependent biases from the reference over sequence edits.

\subsection{The Schr\"odinger Bridge Objective}
Given endpoint distributions $\mu_0$ and $\mu_1$ over discrete sequences, the Schr\"odinger bridge problem seeks a controlled stochastic process whose path measure minimizes relative entropy to a reference process $\mathbb{P}_0$ \citep{schrodinger1932theorie}:
\begin{equation}
\label{eq:sb_problem}
\mathbb{P}_{u^\star}
=
\arg\min_{\mathbb{P}_u \in \mathcal{P}}
\left\{
\mathrm{KL}\!\left(\mathbb{P}_u \,\Vert\, \mathbb{P}_0\right)
\;\middle|\;
(\mathbb{P}_u)_0 = \mu_0,\;
(\mathbb{P}_u)_T = \mu_1
\right\}.
\end{equation}

The solution $\mathbb{P}_{u^\star}$ defines a stochastic interpolation between $\mu_0$ and $\mu_1$. Classical approaches compute this solution via iterative forward--backward projections of the Markov semigroup \citep{vargas2021solving, genevay2018learning}, which are intractable in high-dimensional discrete sequence spaces \citep{sokolov2025exponential}.

\section{Methods}
\label{sec:problem}

We now formalize Minimal-Action Discrete Schrödinger Bridge Matching as a generative transport problem on discrete biological sequences. We consider a reference continuous-time Markov process on sequence space and a controlled process that transports a simple prior distribution to the data distribution along low-action paths. Our objective is to learn a parametric control field $u_{\theta}$ whose induced path measure follows the minimal-action Schr\"odinger bridge connecting the prior and the data.

\subsection{Path Measures and Relative Entropy}
We posit that learning the path measure $\mathbb{P}_u$ induced by the controlled generator $R_u$ is directly possible through a Schr\"odinger bridge objective in Eq.~\eqref{eq:sb_problem}. To explicitly solve this minimization problem, we must quantify the KL divergence between the controlled path measure $\mathbb{P}_{u}$ and the reference $\mathbb{P}_{0}$ induced by the controlled and uninformed generators $R_u$ and $R_0$, respectively. Critically, we derive a tractable form of the path-space KL divergence:
\newline
\begin{theorem}  
(Path-space KL decomposition for CTMCs).
\label{thm:main_kl}
Let $\mathbb{P}_{u}$ and $\mathbb{P}_{0}$ be path measures induced by time-inhomogeneous CTMCs. The relative entropy is the time-integral over instantaneous intensity differences.

\begin{equation}
\label{eq: relative entropy}
\begin{aligned}
\mathrm{KL}(\mathbb{P}_{u} \,\Vert\, \mathbb{P}_{0})
&= \mathbb{E}_{\mathbb{P}_{u}}
\Bigg[
    \int_{0}^{T}
    \sum_{x' \neq X_t}
    \Big(
        R_{u}(X_{t}, x') 
        \log \frac{R_{u}(X_{t}, x')}{R_{0}(X_{t}, x')}
\\
&\hphantom{= \mathbb{E}_{\mathbb{P}_{u}}\Bigg[ \int_{0}^{T} \sum_{x' \neq X_t} \Big(}
        -\, R_{u}(X_{t}, x')
        + R_{0}(X_{t}, x')
    \Big)
    \, dt
\Bigg].
\end{aligned}
\end{equation}
\end{theorem}

The proof is given in Appendix~\ref{path space kl proof}.

\subsection{The Action Functional}
In statistcal physics, the \emph{action} quantifies the accumulated cost of forcing a system away from its equilibrium dynamics. Naturally, we seek the minimal action required to steer the reference process toward the data distribution using the controlled path measure $\mathbb{P}_u$. The cost of this intervention is quantified by the relative entropy between the controlled and reference path measures, which simplifies to the action functional.
\newline
\begin{corollary} [The Action Functional]
By substituting the exponential tilt parameterization, $R_{u}(x, x') = R_{0}(x, x') \exp(u(x, x'))$, into the general relative entropy form (Theorem~\ref{thm:main_kl}), we simplify the path-space relative entropy into the \textbf{Action Functional} $\mathcal{A}(u)$:
\begin{equation}
    \label{eq: action of control}
    \mathcal{A}(u)
    =
    \mathbb{E}_{\mathbb{P}_u}
    \left[
    \int_0^T
    \sum_{x' \neq X_t}
    R_0(X_t,x')\,
    \Psi\!\big(u(X_t,x')\big)\,
    dt
    \right],
\end{equation}
where $\Psi(z) = e^{z} - z - 1$ is the strictly convex cost function associated with the local change of measure.
\end{corollary}

The proof is given in Appendix~\ref{cor:action_tilt}. Consequently, minimal-action transport seeks a control $u^{\star}$ that minimizes $\mathcal{A}(u)$ subject to prescribed endpoint constraints. The resulting process represents the \emph{path of least resistance} that transports the prior to the data distribution with the minimum amount of information injected into the reference dynamics to achieve the desired behavior.

\subsection{Learning Minimal-Action Control Fields}

To minimize the action functional $A(u)$ without solving the full Schrödinger bridge, we seek a time-dependent control field $u_\theta(x,x',t): \mathcal{V}^L \times \mathcal{V}^L \times [0, 1] \to \mathbb{R}^{L \times \mathcal{V}}$  that approximates the optimal control $u^\star$. In MadSBM, this control field can be learned with a neural network with parameters $\theta$, which in turn defines a learned, time-dependent generator for the controlled dynamics:

\begin{equation}
\label{time dependent generator}
R_{u,\theta}(x,x')
=
R_0(x,x')\,\exp\!\big(u_\theta(x,x',t)\big).
\end{equation}

with diagonal entries $R_{\theta, t}(x, x) = - \sum_{x' \ne x} R_{\theta, t}(x, x')$. A similar modification of time-dependency can be made to the cost function $\Psi(u(X_t, x', t))$ used to define the relative entropy between the path measures $\mathbb{P}_u$ and $\mathbb{P}_0$ in Eq.~\eqref{eq: relative entropy} and the action of control $\mathcal{A}(u)
=
\mathbb{E}_{\mathbb{P}_u}
\left[
\int_0^T
\sum_{x' \neq X_t}
R_0(X_t,x')\,
\Psi\!\big(u(X_t,x',t)\big)\,
dt
\right]
$ in Eq.~\eqref{eq: action of control}.

Since our rate matrix is a tilted uniform generator, the time-evolution of this marginal density satisfies the discrete Kolmogorov Forward Equation (KFE). If we let $\rho_{t}^{\theta}$ denote the marginal distribution of $X_{t}$ under $R_{\theta,t}$ the discrete KFE can be written as

\begin{equation}
    \frac{d}{dt} \rho_{t}^{\theta}(x)
    =
    \sum_{y}
    \rho_{t}^{\theta}(y) R_{\theta,t}(y, x)
    -
    \rho_{t}^{\theta}(x) \sum_{y} R_{\theta,t}(x, y).
\end{equation}

Integrating this controlled process backward from $t=1\to 0$ yields a generative procedure analogous to the backward integration of continuous flows. Our formulation overall defines the process $(R_{\theta,t}, \mu_{0})$ as the MadSBM generative model.

\subsection{Training MadSBM}

\paragraph{Optimal Couplings and Interpolation}
To learn a time-dependent control field, MadSBM requires intermediate sequence states $x_t$ interpolating between a noise distribution $\mu_0$ and the data distribution $\mu_1$. We construct such states by pairing fully masked sequences with clean targets and applying a simple masking-based interpolation.

Specifically, we draw a timestep $k \sim \mathrm{Uniform}\{1,\dots,T\}$ and set $t = k/T$. At time $t \in [0,1]$, each token $x_t^{(i)}$ independently reveals the target token $x_1^{(i)}$ with probability $t$ and remains masked with probability $1-t$, defining the forward perturbation kernel
\begin{equation}
\label{eq:masking_rate_sched}
p_t\big(x^{(i)} \mid x_1^{(i)}\big)
=
(1-t)\,\delta_{\mathcal{M}}(x^{(i)})
+
t\,\delta_{x_1^{(i)}}(x^{(i)}),
\end{equation}
where $\delta_{\mathcal{M}}$ is the Kronecker delta on the mask token $\mathcal{M}$. Uniform timestep sampling follows standard practice in discrete diffusion and flow-based models \citep{wang2024diffusion} and provides a tractable approximation to intermediate Schrödinger bridge marginals.

\paragraph{Learning the Control Field}
Given corrupted sequences from Eq.~\eqref{eq:masking_rate_sched}, we learn a control field that tilts a reference process toward the data distribution along minimal-action paths. Under the exponential tilt parameterization, transition rates decompose as
$\log R_u(x,x') = \log R_0(x,x') + u_\theta(x,x')$.
While a uniform random-walk reference $R_0$ is analytically convenient, it fails to capture peptide biophysics and causes trajectories to drift rapidly into low-probability regions.

We therefore define $R_0$ using a biologically informed prior derived from a pre-trained encoder-only protein language model. Specifically, we use logits $f_\phi(x_t) \in \mathbb{R}^{L\times\mathcal{V}}$ from the frozen ESM-2-650M masked language model \citep{lin2023esm2} as reference transition scores. Although ESM-2 is not generative (Appendix~\ref{esm-2 appendix}), its logits provide a local measure of token plausibility suitable for defining reference dynamics. The resulting learned transition distribution decomposes to
\begin{equation}
\label{eq:log_probs}
\log p_\theta(\cdot \mid x_t, t)
\propto
\underbrace{u_\theta(x_t,\cdot,t)}_{\text{Learned tilt}}
+
\underbrace{(1-t)\,f_\phi(x_t)}_{\text{Reference } R_0},
\end{equation}
where the factor $(1-t)$ downweights the reference prior under heavy masking ($t\to 1$), where the language model $f_\phi(\cdot)$ is outside its training distribution.

With Eq.~\eqref{eq:log_probs}, training reduces to maximizing the transition intensity toward the target sequence, equivalent to minimizing the cross-entropy loss on corrupted tokens,
\begin{equation}
\label{eq:madsbm_loss}
\mathcal{L}(\theta)
=
-
\sum_{i=1}^L
\mathbf{1}_{\{x_t^{(i)} \neq x_1^{(i)}\}}
\log p_\theta(x_1^{(i)} \mid x_t,t).
\end{equation}

\paragraph{Analysis of the Training Objective}
Although simple, this objective yields a consistent approximation to minimal-action control, formalized in the following proposition.
\newline
\begin{proposition}[Consistency with Minimal-Action Control]
\label{prop:consistency}
Minimizing the cross-entropy loss in Eq.~\eqref{eq:madsbm_loss} is equivalent to minimizing the KL divergence between the model transition kernel and the optimal forward transitions of the Schrödinger bridge. Consequently, the learned control field $u_\theta$ converges to the unique minimal-action velocity field transporting $\mu_0$ to $\mu_1$ under reference dynamics $R_0$.
\end{proposition}

The proof and additional discussion are provided in Appendix~\ref{thm:loss_consistency}.

\paragraph{Implementation Details}
We parameterize the control field $u_\theta$ using a Diffusion Transformer (DiT) architecture \citep{peebles2023scalable}, enabling time conditioning over ESM-2 latent embeddings. We use a 50M-parameter DiT with a frozen 650M-parameter ESM-2 model to define reference rates, yielding a total inference footprint of approximately 700M parameters while backpropagating gradients only through the 50M parameter control field. Full architectural details, datasets, and training algorithms are provided in Appendices~\ref{dit architecture supplement}, \ref{dataset supplement}, and Algorithm \ref{alg:madsbm_training}.

\subsection{Generative Sampling}
\label{subsec:generation}

\paragraph{Initialization}
After learning the control field $u_\theta$, sampling begins from the fully masked prior $\mu_0 := \{[\texttt{MASK}]\}_{i=1}^L$. We evolve the sequence backward in time from $t = 1 \to 0$ by simulating a CTMC with the learned generator $R_\theta$. In practice, we discretize time into $N$ steps of size $\Delta t = 1/N$.

\paragraph{Transition Intensities}
At each discrete time $t_k = k/N$, we compute the total exit rate from the current sequence
$x_t$ via an exponential tilt of the control field, scaled by a rate coefficient $\lambda = 0.01$:
\begin{equation}
    R_{\text{tot}}(x_t) = \sum_{v \in \mathcal{V}} \exp\left(\lambda \cdot u_\theta(x_t, v, t_k)\right).
\end{equation}

\paragraph{Jump Probabilities}
Next, transitions between CTMC states are governed by a Poisson jump process, a standard result in CTMC theory. The probability that a token updates within the interval $\Delta t$ is thus:
\begin{equation}
    p_{\text{jump}} = 1 - \exp\left(-R_{\text{tot}}(x_t) \cdot \beta \cdot \Delta t\right).
\end{equation}
with $\beta= 0.05$ as a jump scale.

\paragraph{Token Distribution}
Conditioned on a jump occurring, new tokens are sampled from a multinomial distribution ($n=1$) induced by the control field.
Specifically, we define the token distribution produced by the model as
\begin{equation}
\label{eq:token_dist}
p_\theta(v \mid x_t, t_k)
=
\text{Softmax}\!\left(
\text{Top-}p(\frac{u_\theta(x_t, \cdot, t_k)}{\tau}
)\right),
\end{equation}
where $\tau = 0.5$ and nucleus sampling uses $p = 0.9$. Then for each position, we sample a Bernoulli mask $\mathbf{z} \sim \text{Bern}(p_{\text{jump}})$
and update the sequence with new tokens $x_{new}$:

\begin{equation}
\begin{aligned}
x_{\text{new}} 
&\sim \text{Categorical}\!\left(p_\theta(v \mid x_t, t_k)\right) \\
x_{t+\Delta t}
&= \mathbf{z} \odot x_{\text{new}} + (1 - \mathbf{z}) \odot x_t
\end{aligned}
\end{equation}

\paragraph{Objective-Guided Sampling}
\label{guided sampling method}
While unconditional peptide generation enables broad exploration of sequence space, such samples are unlikely to exhibit the binding properties required for therapeutic relevance. Prior work on guiding generative modeling has focused on continuous domains \citep{dhariwal2021diffusion}, with more recent extensions to discrete sequence spaces for biological sequence design \citep{nisonoff2024unlocking, gruver2023protein, vincoff2025soapia}, considering multiple competing objectives \citep{chen2025mogdfm, tang2025peptune}, and optimization of specific sequence tokens \citep{goel2025tokenlevel}. 

As a case-study, we improve the binding affinity of MadSBM-generated peptides by introducing \emph{Objective-Guided Sampling} within our rate-based Schr\"odinger bridge framework. Specifically, we guide the underlying Poisson process using a surrogate binding affinity classifier that quantifies the affinity of a peptide to the target protein. At each timestep, instead of sampling a single $x_{\text{new}}$, we draw $M=16$ candidate sequences
$\{x_{\text{new}}^{(m)}\}_{m=1}^M$ from the categorical $p_\theta(\cdot \mid x_t, t_k)$. Each candidate is then scored using the external classifier model. The resulting affinity scores $\{a^{(m)}\}_{m=1}^M$ are converted into selection weights via $w^{(m)} = \text{Softmax}\!\left(\frac{a^{(m)}}{\tau}\right)$, where $\tau = 0.5$. A single candidate is then sampled according to $w^{(m)}$ and used as $x_{\text{new}}$ in the sequence update.

\paragraph{Convergence of Sampling}
We enforce monotonicity by only allowing transitions for tokens that are currently masked, \textit{i.e.} once a token is unmasked, it cannot be remasked. Under ideal training, $X_{T}^{\theta}$ is distributed close to $\mu_{1}$, and the entire trajectory $(X_{t})_{t \in [0,T]}$ approximates a minimal-action Schrödinger bridge between the prior and data. Additionally, we show that the marginal distributions produced by our sampling procedure converges to the marginals of the minimal-action Schr\"odinger Bridge.
\newline
\begin{proposition}[Convergence of MadSBM Sampling]
\label{prop:sampling_convergence}
The generative procedure defined by the transition probabilities $p_{\text{jump}}$ and multinomial updates constitutes a consistent time-discretization of the controlled CTMC. Specifically, as the step size $\Delta t \to 0$, the marginal distribution of the generated sequences $\widehat{\rho}_t$ converges to the true marginals $\rho_t$ of the minimal-action Schrödinger bridge defined by the generator $R_\theta$.
\end{proposition}
The proof of convergence and full sampling algorithm are given in Appendix~\ref{sampling convergence proof} and in  Algorithm~\ref{alg:madsbm_sampling}.

\section{Results}
\label{sec:results}

\subsection{Unconditional Sequence Generation Quality}
\label{uncond gen results}
\paragraph{Setup}
Unconditional peptide sequence generation broadens the space of potential therapeutic designs. To this end, we use MadSBM to sample 20 sequences for each sequence length $L \in \{5, \dots, 50\}$. For each generated sequence, we evaluate biological plausibility using the ESM-2 pseudo-perplexity (PPL; see Appendix~\ref{esm-2 appendix}) and assess structural stability using the predicted local distance difference test (pLDDT) score \citep{jumper2021af2} of the folded structure produced by ESMFold \citep{lin2023esm2}.

We compare MadSBM against EvoFlow, a state-of-the-art discrete diffusion protein language model (\url{https://huggingface.co/fredzzp/EvoFlow-150M}) built on the Reparameterized Diffusion Model framework (Appendix~\ref{diffusion appendix}) \citep{zheng2023reparameterized}. We elect to compare MadSBM against discrete diffusion models as they have recently become the \textit{de facto} paradigm for biological sequence modeling \citep{wang2024diffusion, gruver2023protein, vincoff2025soapia, goel2025tokenlevel, tang2025peptune}, outperforming traditional autoregressive baselines \citep{nijkamp2023progen2, ferruz2022protgpt2}. Since EvoFlow is trained on a subset of UniProt sequences and peptides are biologically considered to be short proteins, we generate peptide samples with EvoFlow by simply restricting sequence lengths to the range of 5–50 residues. EvoFlow samples new peptide sequences using the Path-Planning scheme (self-planner variant) introduced by \citeauthor{peng2025path} To ensure a fair comparison, we use the 150M-parameter EvoFlow model when benchmarking against our 50M-parameter DiT-based MadSBM. Each model is evaluated at various sampling step budgets, $N \in \{32, 64, 128 \}$.

\paragraph{Results}
Both MadSBM and the discrete diffusion (DD) baseline generate peptide sequences with lower average PPLs than the held-out test set, which is a proxy for the natural peptide sequence distribution. Across the sampling budgets, MadSBM produces sequences with lower or competitive PPLs than the DD baseline, indicating higher biological alignment to the peptide sequence space (Table \ref{tab:uncond_gen_quality}).

\newpage
\begin{wraptable}{r}{0.55\textwidth}  
  \vspace{-0.5em}
  \centering
  \begin{minipage}{\linewidth}
    \captionsetup{skip=4pt}
    \captionof{table}{\textbf{Unconditional sequence generation quality across varying sampling step budgets ($N$).}
    For results, the mean across all sequence lengths and standard deviation are reported.
    MadSBM is compared against the discrete diffusion (DD) model EvoFlow.}
    \label{tab:uncond_gen_quality}
    \medskip
    
    \scriptsize
    \renewcommand{\arraystretch}{1.05} 

    \centering
    \begin{tabular}{@{} ll cc @{}}
      \toprule
      Steps & Model & PPL $(\downarrow)$ & pLDDT $(\uparrow)$ \\
      \midrule
      \multirow{2}{*}{$N=32$}
        & DD       & 10.990 {\scriptsize $\pm$ 6.766} & 71.608 {\scriptsize $\pm$ 9.692} \\
        & \textbf{MadSBM} & \textbf{8.389 {\scriptsize $\pm$ 10.873}} & \textbf{71.687 {\scriptsize $\pm$ 11.835}} \\
      \midrule
      \multirow{2}{*}{$N=64$}
        & DD       & 9.042 {\scriptsize $\pm$ 4.679} & \textbf{73.848 {\scriptsize $\pm$ 9.436}} \\
        & MadSBM   & \textbf{8.943 {\scriptsize $\pm$ 15.384}} & 71.604 {\scriptsize $\pm$ 12.223} \\
      \midrule
      \multirow{2}{*}{$N=128$}
        & \textbf{DD}     & \textbf{7.617 {\scriptsize $\pm$ 6.834}} & \textbf{75.784 {\scriptsize $\pm$ 8.787}} \\
        & MadSBM   & 8.719 {\scriptsize $\pm$ 12.925} & 70.725 {\scriptsize $\pm$ 12.041} \\
      \midrule
      \multicolumn{2}{l}{\textbf{Test Set}} 
        & 13.385 {\scriptsize $\pm$ 5.524} 
        & 63.273 {\scriptsize $\pm$ 11.710} \\
      \bottomrule
    \end{tabular}
  \end{minipage}
  \vspace{-1em}
\end{wraptable}

Notably, the DD baseline produces sequences that fold into higher-confidence structures, as reflected by its higher average pLDDT scores. However, it is important to note that folding models used to compute metrics such as pLDDT rely on evolutionary information for predictions, which can obscure these models' ability to directly assess sequence–structure compatibility \citep{korbeld2025limitations}.

Additionally, prior work has shown that diffusion models that reverse a predefined corruption process require sufficiently large sampling steps to maintain sample quality, as coarse discretizations that do not form a smooth transition from noise to data distributions result in large jumps in the token space \citep{xue2024accelerating, shih2023parallel}. Notably, MadSBM outperforms the DD baseline even under low sampling budgets: at $N=32$, MadSBM achieves lower PPL while maintaining competitive pLDDT scores.

Overall, these results indicate that MadSBM improves on current state-of-the-art DD models in unconditionally generating biologically relevant peptide sequences in a sample-efficient manner.

\subsection{Navigating Probability Paths in the Discrete Sequence Space}

\paragraph{Setup}
Discretized sampling processes in continuous and discrete generative models often force the evolving sequence through low-probability or biologically invalid regions of the sample space, potentially degrading final sample quality \citep{sahoo2024mdlm, domingo2024adjoint, xue2024accelerating, shih2023parallel}. Accordingly, we analyze the likelihood evolution of sequences generated by MadSBM and baseline models  to assess adherence to the protein manifold during generation. During the unconditional sampling procedure described in Section~\ref{subsec:generation}, we compute the negative log-likelihood (NLL) of the intermediate states $x_t$ using the ESM-2 model at each timestep. This metric serves as a proxy for biological plausibility, as $e^{-\text{NLL}} \propto \text{PPL}$. The NLL for all sequence lengths at the given sampling step were averaged.

\paragraph{Results}
We compare the NLL of different models across sampling iterations, using this metric as a proxy of the overall sequence likelihood across the generative trajectory. Figure~\ref{fig:prob_paths} shows that the DD baseline follows a constrained and low-variance trajectory, while MadSBM exhibits greater path diversity and variance, reflecting its ability to explore a broader set of stochastic transport paths. Specifically, the DD baseline is confined to a narrow likelihood window throughout sampling, whereas MadSBM exhibits more variance of likelihoods, enabling it to reach higher-likelihood regions of the sequence manifold earlier in the generation process. Additionally, while MadSBM exhibits a lower likelihood floor during sampling compared to DD, which could result in exploration of a wider and potentially lower-quality region of the sequence space, this does not compromise generative quality as the model ultimately converges to superior or competitive sequence likelihoods (Table~\ref{tab:uncond_gen_quality}).

\begin{figure}[h]
    \centering
    \includegraphics[width=0.8\linewidth]{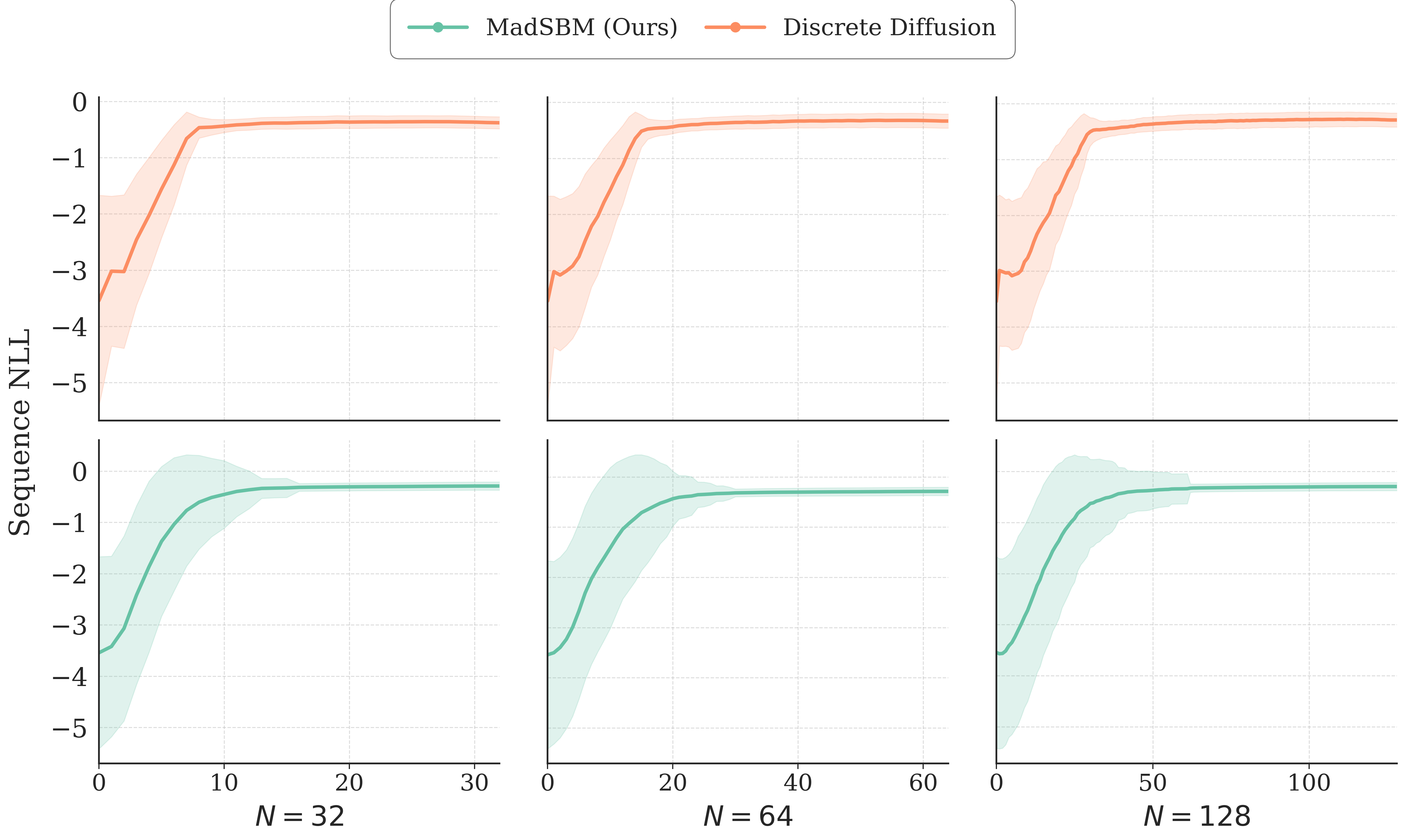}
    \caption{\textbf{Probability paths taken by models under various sampling budgets} $(N)$. The y-axis represents the NLL of the sequence at the current iteration, assessed by the ESM-2-650M protein language model. The shaded area around the traced trajectory represents the standard deviation of the NLL at the currrent sampling iteration.}
    \label{fig:prob_paths}
\end{figure}

\subsection{Ablation on Reference Dynamics}
\paragraph{Setup}
Recall that MadSBM constructs its generative distribution by applying a learned tilt to the biologically-inspired reference process, which is defined through the logits produced by the ESM-2 model. To assess the benefit of using this biologically-relevant reference, we ablate the use of ESM-2 during training by 1) removing the time-gating mechanism and 2) replacing the ESM-derived reference with a random walk over sequence tokens modeled by a uniform generator (Eq.~\eqref{eq:uninformed_ctmc}) that does not incorporate any biophysical peptide representations.

\paragraph{Results} Table~\ref{tab:esm2_ablation} shows that ablating the gating mechanism and ESM-2 itself leads to an increase in perplexity when evaluating the model on the held-out test set, indicating a degradation in generative quality when the biophysical reference dynamics is removed.

\begin{wraptable}{l}{0.45\textwidth} 
  \vspace{-0.75\baselineskip}       

  \begin{minipage}{\linewidth}      
    \captionsetup{skip=4pt}          
    \captionof{table}{\textbf{Ablation of principled reference process components in MadSBM}. We train MadSBM models without time-gating and without any ESM dependency and assess the resulting perplexity on the held-out test set.}
    \label{tab:esm2_ablation}

    \medskip

    \centering
    \scriptsize                       

    \setlength{\tabcolsep}{4pt}      

    \renewcommand{\arraystretch}{0.95}

    \begin{tabular}{@{} l c c @{}}
      \toprule
      & $\log R_0$ & Test PPL ($\downarrow$) \\
      \midrule
      \textbf{MadSBM} 
        & $(1-t)f_\phi(x_t)$ 
        & \textbf{4.503} \\
      \quad w/o gating 
        & $f_\phi(x_t)$ 
        & 4.987 \\
      \qquad \& w/o ESM-2 
        & $0$ 
        & 4.750 \\
      \bottomrule
    \end{tabular}

  \end{minipage}
  \vspace{-0.5\baselineskip} 
\end{wraptable}

Interestingly, fully ablating ESM-2 and the time-gating achieves a lower PPL, outperforming the variant that solely removes time-gating. While initially counterintuitive, this behavior is explained by the training distribution of ESM-2: the model is trained under a masked language modeling objective in which only 15\% of tokens are replaced by the \texttt{[MASK]} token (Appendix~\ref{esm-2 appendix}). As a result, ESM-2 representations become increasingly unreliable at higher masking rates, where the model operates out of distribution. Removing the time-gating mechanism forces the model to rely on ESM logits even in these high-masking regimes, leading to noisier reference dynamics and poorer performance. In contrast, fully ablating ESM-2 eliminates this mismatch entirely, requiring the DiT backbone to jointly encode the peptide biophysical properties while learning the optimal tilt. These results overall validate our choice of using a principled, time-controlled reference process that appropriately modulates the influence of the biological prior across the corruption trajectory.

\subsection{Minimal-Action Sampling Trajectories}

\paragraph{Setup}
A central idea of MadSBM is that minimizing a simple cross-entropy objective corresponds to learning low-action transport paths between a simple reference process and the data distribution. To empirically evaluate this claim, we measure $\mathcal{A}_L(u^w)$, the \textit{instantaneous actional} (the actional at each sampling step) as a proxy for the control cost incurred by our learned tilt.

\begin{wraptable}{r}{0.50\textwidth} 
  \vspace{-0.75\baselineskip}        

  \begin{minipage}{\linewidth}       
    \captionof{table}{\textbf{Instantaneous worst-case actionals for different sampling budget discretizations} ($\Delta t$).
    We report the maximum observed logit $M$ and the corresponding instantaneous actional evaluated for each model.}
    \label{tab:instantaneous_actionals}

    \medskip

    \centering
    \small

    \begin{tabularx}{\linewidth}{@{} c X c c @{}}
      \toprule
      \(\Delta t\) & Model & \(M\) & $\mathcal{A}_L(u^w)$ \\
      \midrule
      \multirow{2}{*}{$1/32$}
        & MadSBM    & 15.308787 & $1.39114\times10^{5}$ \\
        & w/o ESM-2 & 14.275414 & $4.94968\times10^{4}$ \\
      \midrule
      \multirow{2}{*}{$1/64$}
        & MadSBM    & 15.308787 & $6.95569\times10^{4}$ \\
        & w/o ESM-2 & 14.275414 & $2.47484\times10^{4}$ \\
      \midrule
      \multirow{2}{*}{$1/128$}
        & MadSBM    & 15.308787 & $3.47784\times10^{4}$ \\
        & w/o ESM-2 & 14.275414 & $1.23742\times10^{4}$ \\
      \bottomrule
    \end{tabularx}
  \end{minipage}

\end{wraptable}

 While sampling the same sequences used in the unconditional sequence generation and probability path evaluations, we computed the worst-case instantaneous actionals for various sampling budgets using the maximum logit value $M$. This worst-case scenario serves as the upper bound on the actionals that MadSBM should produce. We encourage the reader to review the derivation of the instantaneous actional in Appendix~\ref{instantaneous actionals appendix}.

\paragraph{Results}
Table~\ref{tab:instantaneous_actionals} reports the maximum observed logits $M$ and the corresponding worst-case actions $\mathcal{A}_L(u^w)$, and Figure ~\ref{fig:action_trajs} records the instantaneous actionals produced by MadSBM and its ablated counterpart models.

\begin{figure} \centering \includegraphics[width=1.0\linewidth]{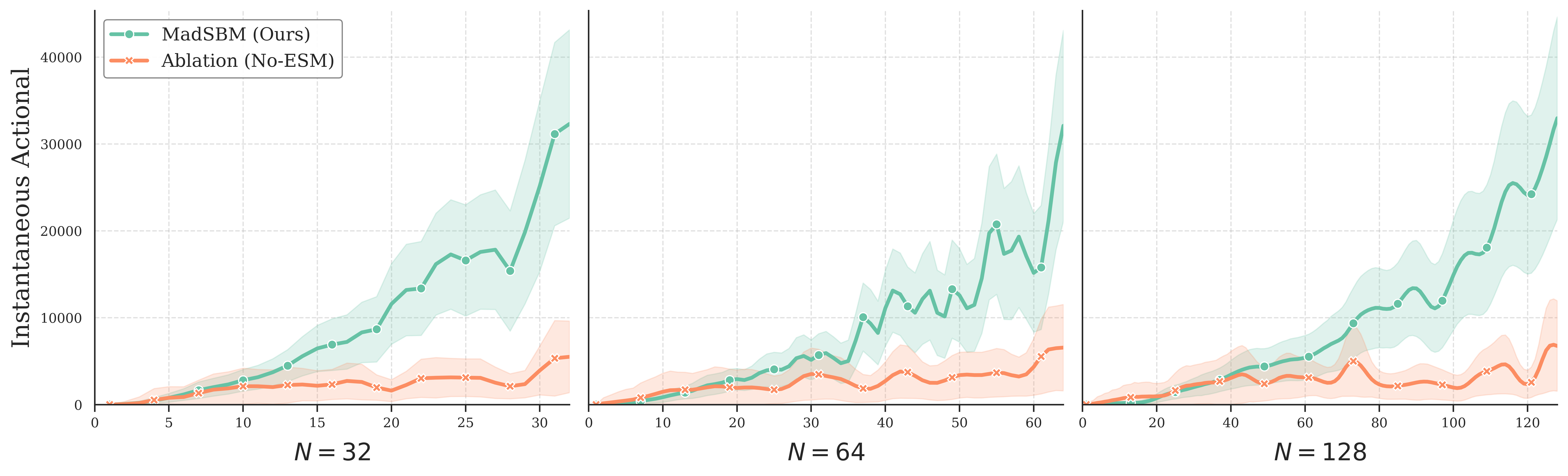} \caption{\textbf{Instantaneous actional values for MadSBM and ESM-ablated counterpart}. The y-axis represents the actional $\mathcal{A}_L(u)$ at the current timestep. Results are shown at various sampling budgets.} 
    \label{fig:action_trajs}
\end{figure}

Figure~\ref{fig:action_trajs} shows that the instantaneous actionals incurred by MadSBM remain comfortably below the corresponding worst-case actionals in Table~\ref{tab:instantaneous_actionals} throughout the sampling trajectory. This behavior is consistent with our theoretical framing of MadSBM learning minimal-action transport paths: although biologically meaningful, the reference process can be noisy, requiring greater action from the learned control field to tilt toward the data distribution at each timestep. In contrast, the ESM-ablated model exhibits near-constant actionals across sampling steps, which is also expected given the control field must not learn to adjust distorted peptide sequence representations from the reference process. Interestingly, we observe that MadSBM actionals increase as the sampling process converges. This reflects the fact that later generative steps must commit high-confidence tokens to the sequence, which further aligns with the low NLL variance at later sampling steps in Figure.~\ref{fig:prob_paths}.

\subsection{Binding Affinity Optimization}
\paragraph{Setup}
We evaluate MadSBM’s ability to design valid peptides by generating candidate binders for four disease targets with known binders and one target with no known binders. For each target, we generate 60 peptide sequences at various sequence lengths  by concatenating an $L$-length sequence of \texttt{[MASK]} tokens to the target amino acid sequence and running the MadSBM sampling procedure. To enhance binding affinity, we subsequently resample 60 peptides per target using objective-guided sampling (Section~\ref{guided sampling method}), where the guidance signal is provided by a pre-trained binding affinity predictor. Specifically, we use the unpooled wild-type to wild-type binding affinity model from PeptiVerse \citep{zhang2026peptiverse}.  In the unconditional and guidance cases, we use $N=32$ sampling steps to highlight MadSBM's sample step efficiency and sequence lengths $L \in \{10, 15, 20\}$ to align with the length distribution of experimentally characterized peptide binders.

\paragraph{Results} Table~~\ref{tab:binding_affinities} shows the binding affinities, ipTM scores, and AutoDock VINA scores for four targets with existing binderse and for a target without an existing binder. 

\begin{table}[h]
\centering
\renewcommand{\arraystretch}{0.95} 
\setlength{\tabcolsep}{4pt}        
\caption{\textbf{Binding affinity scores for designed and existing binders.} Existing binder affinities and sequences are taken from \citep{tang2025gumbelfm}. ipTM values were sourced from AlphaFold3 and docking scores were sourced from AutoDock VINA. Values are reported as average across 10, 15, and 20-length generated peptides. The targets 5E1C, 4EZN, 1AYC, and 5KRI have existing binders, while 3HVE has no existing binders.}
\label{tab:binding_affinities}
\scriptsize
\resizebox{\textwidth}{!}{%
\begin{tabular}{lccccccccc}
\toprule
 & \multicolumn{3}{c}{\textbf{Binding Affinity} ($\uparrow$)}
 & \multicolumn{3}{c}{\textbf{Best ipTM} ($\uparrow$)}
 & \multicolumn{3}{c}{\textbf{Docking Score} (kcal/mol) ($\downarrow$)} \\
\cmidrule(lr){2-4} \cmidrule(lr){5-7} \cmidrule(lr){8-10}
Target 
& Existing & Unconditional & Guided 
& Existing & Unconditional & Guided
& Existing & Unconditional & Guided \\
\midrule

5E1C  & 4.932 & 5.416 & 6.109 & 0.83 & 0.04 & 0.75 & -4.3 & -6.0 & -6.3 \\
4EZN  & 6.176 & 5.468 & 6.072 & 0.53 & 0.28 & 0.68 & -4.1 & -6.0 & -6.8 \\
1AYC  & 6.576 & 7.272 & 7.982 & 0.58 & 0.34 & 0.57 & -5.3 & -6.8 & -7.5 \\
5KRI  & 4.932 & 5.416 & 6.109 & 0.83 & 0.05 & 0.76 & -3.5 & -6.5 & -6.7 \\
\midrule
3HVE  
& \textbf{--} & 5.372 & 6.108 
& \textbf{--} & 0.09 & 0.37
& \textbf{--} & -6.3 & -6.8 \\
\bottomrule
\end{tabular}%
} 
\end{table}

Notably, for three of the four targets with binders, MadSBM-generated peptides \textit{without guidance} have greater binding affinities and docking scores than existing peptides, indicating that MadSBM learned a therapeutically potent peptide distribution. Applying guidance further improves upon binding affinities and docking scores, demonstrating, to our knowledge, the first-ever application of discrete classifier guidance \citep{gruver2023protein, nisonoff2024unlocking, goel2025tokenlevel, tang2025peptune, vincoff2025soapia} to a Schr\"odinger bridge matching-based generative model. We further highlight that MadSBM was constrained to a low sampling budget relative to what's required by state-of-the-art discrete diffusion models, and that while existing binders are often manually curated through slow experimental screening, MadSBM rapidly produces high-quality candidates through principled generative modeling.

\section*{Discussion}
In this work, we introduce \textbf{MadSBM}, which reframes peptide sequence generation as a rate-based stochastic transport problem, showing that low-action Schrödinger bridge dynamics can be learned directly in discrete sequence spaces when coupled with an informative biological reference process. By defining generation relative to pre-trained protein language model logits and learning a time-dependent control field, MadSBM produces probability paths that remain close to high-likelihood peptide neighborhoods throughout sampling, leading to improved perplexity and competitive structural confidence at substantially lower sampling budgets than discrete diffusion and flow-based baselines. At the same time, the method inherits limitations from its reference dynamics, including reliance on pre-trained language models that may be poorly calibrated under extreme masking and sensitivity to the choice of time-gating and rate scaling. Future work will focus on experimentally validating generated peptides, refining chemistry- and structure-aware reference processes, extending the framework to larger protein domains, integrating higher-fidelity physical or experimental feedback into the control field, and extensions to other sequence domains (e.g. natural language, nucleic acids). More broadly, MadSBM opens a path toward discrete generative models that expose and control the entire probability trajectory of generation rather than only the endpoint distribution.

\section*{Impact Statement}
MadSBM extends current discrete generative modeling frameworks by introducing a rate-based Schrödinger bridge framework that enables sample-efficient and controllable peptide sequence generation. By improving the reliability of both unconditional and guided peptide design, this approach has the potential to accelerate early-stage therapeutic discovery while making the generative process more interpretable and auditable. To mitigate misuse, we focus on short peptide sequences, rely on pre-trained biological priors, and emphasize downstream validation and safety-aware filtering before any experimental deployment.

\newpage
\bibliography{pbg_template}

\newpage
\onecolumn
\appendix

\begin{center}
    {\LARGE \bfseries Appendix}
\end{center}
\vspace{1em}

\counterwithin{figure}{section}
\renewcommand{\thefigure}{\thesection\arabic{figure}}
\setcounter{figure}{0}

\section{Theoretical Proofs}
\label{sec:proofs}

In this section we provide theoretical guarantees for Minimal-Action Schr\"odinger Bridges (MadSBM). We first derive a path-space relative entropy identity for controlled continuous-time Markov chains (CTMCs) on discrete sequence spaces. We then characterize the discrete Schr\"odinger bridge as an endpoint-tilted reference path measure and show that its optimal generator is a Doob $h$-transform of the reference generator. This yields a discrete Hamilton--Jacobi--Bellman (HJB) equation for the optimal log-potential and an explicit expression for the optimal edge control. Finally, we analyze the cross-entropy objective, provide stability bounds that relate control error to path-space divergence, establish representational completeness of the exponential-tilt parameterization, and prove convergence of the time-discretized sampler.

\paragraph{Notation} We adopt the notation from the main text. Let $\mathcal{V}$ denote the finite alphabet and $\mathcal{X} := \mathcal{V}^{L}$ as the length-$L$ sequence space. We set $\mu_{0}$ as the simple prior on $\mathcal{X}$ and $\mu_{1}$ as the data distribution on $\mathcal{X}$. For our dynamics, let $R_{0}$ be the reference generator on $\mathcal{X}$ with $R_{0}(x,x') \ge 0$ for $x' \ne x$ and $R_{0}(x,x) = -\sum_{y \ne x} R_{0}(x,y)$, let $R_t$ be the time-dependent controlled generator with induced path law on $[0,T]$ denoted by $\mathbb{P}$, and let the reference path law induced by $(\mu_{0}, R_{0})$ be denoted by $\mathbb{P}_{0}$. We define the multiplicative control parameterization as $R_{t}(x,x') = R_{0}(x,x')\, a_{t}(x,x')$ with $a_{t}(x,x') > 0 \text{ for } x' \ne x$ and sometimes use $a_{t}(x,x') = \exp(u_{t}(x,x'))$ with log-control $u_{t}(x,x') \in \mathbb{R}$. The Schr\"odinger bridge path law between $(\mu_{0},\mu_{1})$ relative to $\mathbb{P}_{0}$ is denoted as $\mathbb{P}^{\star}$, with generator $R^{\star}_{t}$.








\subsection{Path-Space Relative Entropy for Controlled CTMCs}

\label{path space kl proof}

We derive the standard relative-entropy identity for CTMC path measures. This result is the discrete analogue of Girsanov-type identities for diffusion processes.

\begin{theorem}[Path-space KL decomposition for CTMCs]
\label{thm:ctmc_kl}
Let $\mathbb{P}$ and $\mathbb{P}_{0}$ be path measures on $D([0,T];\mathcal{X})$ induced by time-inhomogeneous CTMCs with generators $R_{t}$ and $R_{0}$ and initial distributions $\nu$ and $\nu_{0}$, respectively. Assume absolute continuity in the sense that for all $t \in [0,T]$ and $x \ne x'$,
\begin{equation}
    R_{t}(x,x') > 0 \;\Rightarrow\; R_{0}(x,x') > 0.
\end{equation}
Then $\mathbb{P} \ll \mathbb{P}_{0}$ and the path-space relative entropy satisfies
\begin{equation}
\label{eq:ctmc_kl_density}
\begin{aligned}
\mathrm{KL}(\mathbb{P}_{u} \,\Vert\, \mathbb{P}_{0})
=
\mathbb{E}_{\mathbb{P}_{u}}
\Bigg[
\int_{0}^{T}
\sum_{x' \neq X_t}
\Big(
R_{u}(X_{t}, x') \log \tfrac{R_{u}(X_{t}, x')}{R_{0}(X_{t}, x')}
- R_{u}(X_{t}, x') + R_{0}(X_{t}, x')
\Big)
\, dt
\Bigg].
\end{aligned}
\end{equation}
\end{theorem}

\begin{proof}
A sample path of a CTMC on a finite state space can be described by an initial state
$X_{0}=x_{0}$ having $N$ jumps with jump times
$0 < \tau_{1} < \cdots < \tau_{N} \le T$, and post-jump states
$x_{1},\dots,x_{N}$, where $x_{i} \ne x_{i-1}$ for each $i$ and
$X_{t}=x_{i}$ for $t \in [\tau_{i},\tau_{i+1})$ with $\tau_{0}=0$ and
$\tau_{N+1}=T$. Let the exit rates be defined as
\begin{equation}
    \lambda_{t}(x) := \sum_{y \ne x} R_{t}(x,y),
    \qquad
    \lambda_{0}(x) := \sum_{y \ne x} R_{0}(x,y).
\end{equation}

We now compute the likelihood of a realized path under the MadSBM model specified by
$(\nu,R_{t})$. The probability of starting in state $x_{0}$ is given by
\begin{equation}
    \mathbb{P}(X_{0}=x_{0}) = \nu(x_{0}).
\end{equation}

While the process remains in state $x_{i}$, no jump occurs on the interval
$[\tau_{i},\tau_{i+1})$. For a time-inhomogeneous CTMC, the corresponding
survival probability is
\begin{equation}
    \exp\!\left(
        -\int_{\tau_{i}}^{\tau_{i+1}} \lambda_{t}(x_{i})\,dt
    \right).
\end{equation}
Multiplying over all inter-jump intervals yields the total survival contribution
\begin{equation}
    \prod_{i=0}^{N}
    \exp\!\left(
        -\int_{\tau_{i}}^{\tau_{i+1}} \lambda_{t}(x_{i})\,dt
    \right)
    =
    \exp\!\left(
        -\sum_{i=0}^{N}
        \int_{\tau_{i}}^{\tau_{i+1}} \lambda_{t}(x_{i})\,dt
    \right).
\end{equation}

At each jump time $\tau_{i}$, the instantaneous probability density of
transitioning from $x_{i-1}$ to $x_{i}$ is given by the transition rate
$R_{\tau_{i}}(x_{i-1},x_{i})$. The contribution from all jumps is therefore given by the product
\begin{equation}
    \prod_{i=1}^{N} R_{\tau_{i}}(x_{i-1},x_{i}).
\end{equation}

Combining the initial distribution, survival probabilities, and jump
intensities, the likelihood of the path under $\mathbb{P}$ is
\begin{equation}
    p_{\mathbb{P}}
    =
    \nu(x_{0})
    \left(
        \prod_{i=1}^{N} R_{\tau_{i}}(x_{i-1},x_{i})
    \right)
    \exp\!\left(
        -\sum_{i=0}^{N}
        \int_{\tau_{i}}^{\tau_{i+1}} \lambda_{t}(x_{i})\, dt
    \right).
\end{equation}
A completely analogous expression holds for $p_{\mathbb{P}_{0}}$ with
$(\nu_{0},R_{0})$.

Since we aim to find the KL divergence, we will take the ratio and the logarithm. Using the support assumption yields
\begin{align}
    \log\frac{p_{\mathbb{P}}}{p_{\mathbb{P}_{0}}}
    &=
    \log\frac{\nu(X_{0})}{\nu_{0}(X_{0})}
    +
    \sum_{i=1}^{N} \log\frac{R_{\tau_{i}}(X_{\tau_{i}^{-}},X_{\tau_{i}})}{R_{0}(X_{\tau_{i}^{-}},X_{\tau_{i}})}
    -
    \int_{0}^{T} \left(\lambda_{t}(X_{t}) - \lambda_{0}(X_{t})\right)\, dt.
    \label{eq:rn_derivative_ctmc}
\end{align}

Taking expectation with respect to $\mathbb{P}$ yields
\begin{equation}
\begin{aligned}
\mathrm{KL}(\mathbb{P}\,\|\,\mathbb{P}_0)
&=
\mathbb{E}_{\mathbb{P}}
\left[
\log\frac{p_{\mathbb{P}}}{p_{\mathbb{P}_{0}}}
\right] \\
&=
\mathbb{E_P}\Bigg[\textrm{log} \big ( \frac{v(X_0)}{v_0(X_0} \big ) \Bigg]
+
\mathbb{E}_{\mathbb{P}}
\left[
\sum_{x' \neq X_t}
\log
\frac{R_t(X_t,x')}{R_0(X_t,x')}
\,dt
\right]
-
\mathbb{E}_{\mathbb{P}}
\left[
\int_{0}^{T} \left(\lambda_{t}(X_{t}) - \lambda_{0}(X_{t})\right)\, dt
\right].
\end{aligned}
\label{eq:path_kl_ctmc}
\end{equation}

Now we simplify. The first term yields $\mathrm{KL}(\nu\Vert\nu_{0})$ by definition and vanishes to zero. For the second term (sum over jumps), we use the Martingale compensator identity for CTMCs, which states for any bounded measurable function $f(t,x,x')$:

\begin{equation}
    \mathbb{E}_{\mathbb{P}}
    \left[
        \sum_{i=1}^{N} f(\tau_{i},X_{\tau_{i}^{-}},X_{\tau_{i}})
    \right]
    =
    \mathbb{E}_{\mathbb{P}}
    \left[
        \int_{0}^{T}
        \sum_{x' \ne X_{t}} R_{t}(X_{t},x') f(t,X_{t},x')\, dt
    \right].
\end{equation}

Applying this with $f(t,x,x') := \log\frac{R_{t}(x,x')}{R_{0}(x,x')}$ transforms the sum over jumps into an integral over time, giving:

\begin{equation}
    \mathbb{E_P} \Bigg[ \int_0^T \sum_{x' \neq xX_t} R_t(X_t, x')\log \frac{R_t(X_t, x')}{R_0(X_t, x')} dt \Bigg] 
\end{equation}

For the third and final term, note that

\begin{equation}
    \lambda_{t}(X_{t}) - \lambda_{0}(X_{t})
    =
    \sum_{x'\ne X_{t}} \left( R_{t}(X_{t},x') - R_{0}(X_{t},x') \right)
\end{equation}

Substituting these identities into \eqref{eq:rn_derivative_ctmc} yields \eqref{eq:ctmc_kl_density}.
\end{proof}

\begin{corollary}[Action form under exponential tilting]
\label{cor:action_tilt}
Assume $\nu=\nu_{0}=\mu_{0}$ and $R_{t}(x,x') = R_{0}(x,x') \exp(u_{t}(x,x'))$ for $x'\ne x$. Then
\begin{equation}
    \mathrm{KL}(\mathbb{P}\,\Vert\,\mathbb{P}_{0})
    =
    \mathbb{E}_{\mathbb{P}}
    \left[
        \int_{0}^{T}
        \sum_{x' \ne X_{t}}
        R_{0}(X_{t},x')
        \left(
            e^{u_{t}(X_{t},x')} u_{t}(X_{t},x')
            -
            e^{u_{t}(X_{t},x')}
            +
            1
        \right)
        dt
    \right].
    \label{eq:action_exptilt}
\end{equation}
\end{corollary}

\begin{proof}
Apply Theorem~\ref{thm:ctmc_kl} with $R_{t}=R_{0}e^{u_{t}}$ and cost function $\Psi(z) = e^z-z-1$. For each edge,
\begin{equation}
\begin{aligned}
R_t\log\frac{R_t}{R_0} - R_t + R_0
&= R_t\left(\log\frac{R_t}{R_0}\right) - R_t + R_0 \\
&= R_t\cdot u - R_t + R_0 &&\text{(since }\log\frac{R_t}{R_0}=u\text{)}\\
&= (R_0 e^{u})\,u - R_0 e^{u} + R_0 &&\text{(since }R_t=R_0 e^{u}\text{)}\\
&= R_0\bigl(e^{u}u - e^{u} + 1\bigr).
\end{aligned}
\end{equation}
and summing over $x'\ne X_{t}$ yields Eq. \eqref{eq:action_exptilt}.
\end{proof}

\subsection{Discrete Schr\"odinger Bridge Structure}

We now characterize the Schr\"odinger bridge as the KL projection of the reference path law onto the set of path laws with prescribed endpoint marginals.

\begin{theorem}[Endpoint-tilted form and uniqueness]
\label{thm:sb_endpoint_tilt}
Let $\mathcal{C}$ denote the set of path measures on $D([0,T];\mathcal{X})$ whose endpoint marginals satisfy $X_{0}\sim\mu_{0}$ and $X_{T}\sim\mu_{1}$. Assume $\mathcal{C}$ is nonempty and that there exists at least one $\mathbb{Q}\in\mathcal{C}$ with $\mathrm{KL}(\mathbb{Q}\Vert \mathbb{P}_{0})<\infty$. Then the optimization problem
\begin{equation}
    \mathbb{P}^{\star}
    \in
    \arg\min_{\mathbb{Q}\in\mathcal{C}}
    \mathrm{KL}(\mathbb{Q}\,\Vert\,\mathbb{P}_{0})
    \label{eq:sb_primal}
\end{equation}
admits a unique minimizer $\mathbb{P}^{\star}$. Moreover, there exist functions $f,g:\mathcal{X}\to(0,\infty)$ such that
\begin{equation}
    \frac{d\mathbb{P}^{\star}}{d\mathbb{P}_{0}}(X)
    =
    \frac{f(X_{0})\, g(X_{T})}{\mathbb{E}_{\mathbb{P}_{0}}[f(X_{0}) g(X_{T})]}.
    \label{eq:sb_factor_form}
\end{equation}
\end{theorem}

\begin{proof}
\textbf{Uniqueness.} The feasible set $\mathcal{C}$ is convex because endpoint marginal constraints are linear in $\mathbb{Q}$. The functional $\mathbb{Q}\mapsto \mathrm{KL}(\mathbb{Q}\Vert \mathbb{P}_{0})$ is strictly convex on $\{\mathbb{Q}: \mathbb{Q}\ll \mathbb{P}_{0}\}$, hence the minimizer is unique.

\textbf{Factor form.} Consider the constrained minimization of $\mathrm{KL}(\mathbb{Q}\Vert \mathbb{P}_{0})$ over $\mathbb{Q}\ll\mathbb{P}_{0}$ with endpoint constraints. Introduce Lagrange multipliers $\alpha,\beta:\mathcal{X}\to\mathbb{R}$ for the constraints $\mathbb{Q}(X_{0}=x)=\mu_{0}(x)$ and $\mathbb{Q}(X_{T}=x)=\mu_{1}(x)$. The Lagrangian is
\begin{equation}
    \mathcal{L}(\mathbb{Q},\alpha,\beta)
    =
    \mathrm{KL}(\mathbb{Q}\Vert\mathbb{P}_{0})
    +
    \sum_{x\in\mathcal{X}} \alpha(x)\left(\mu_{0}(x)-\mathbb{Q}(X_{0}=x)\right)
    +
    \sum_{x\in\mathcal{X}} \beta(x)\left(\mu_{1}(x)-\mathbb{Q}(X_{T}=x)\right).
\end{equation}
A standard variational argument on densities $d\mathbb{Q}/d\mathbb{P}_{0}$ yields the stationarity condition
\begin{equation}
    \log\frac{d\mathbb{Q}}{d\mathbb{P}_{0}}(X)
    =
    \alpha(X_{0}) + \beta(X_{T}) - c
\end{equation}
for some constant $c$. Exponentiating gives
\begin{equation}
    \frac{d\mathbb{Q}}{d\mathbb{P}_{0}}(X)
    \propto
    e^{\alpha(X_{0})} e^{\beta(X_{T})}.
\end{equation}
Setting $f(x):=e^{\alpha(x)}$ and $g(x):=e^{\beta(x)}$ yields \eqref{eq:sb_factor_form} after normalization. The functions $f,g$ are then selected to satisfy the endpoint constraints, which is possible by assumption that $\mathcal{C}$ is nonempty and a finite-KL feasible point exists.
\end{proof}

\subsection{Doob $h$-Transform and Optimal Generator}

The endpoint tilt implies that $\mathbb{P}^{\star}$ is Markov and admits an explicit generator as a Doob $h$-transform of $R_{0}$.

\begin{theorem}[Doob transform form of the Schr\"odinger bridge]
\label{thm:doob_transform}
Let $\mathbb{P}^{\star}$ be as in Theorem~\ref{thm:sb_endpoint_tilt}. Define the backward potential
\begin{equation}
    h_{t}(x)
    :=
    \mathbb{E}_{\mathbb{P}_{0}}\!\left[ g(X_{T}) \,\middle|\, X_{t}=x \right],
    \qquad t \in [0,T],
    \label{eq:h_potential}
\end{equation}
where $g$ is the endpoint factor from \eqref{eq:sb_factor_form}. Then $h_{t}(x)>0$ and the optimal bridge $\mathbb{P}^{\star}$ is a time-inhomogeneous Markov process with generator
\begin{equation}
    R^{\star}_{t}(x,x')
    =
    R_{0}(x,x') \frac{h_{t}(x')}{h_{t}(x)},
    \qquad x'\ne x,
    \label{eq:doob_generator}
\end{equation}
and diagonal entries $R^{\star}_{t}(x,x) = -\sum_{y\ne x} R^{\star}_{t}(x,y)$.
Equivalently, in log-control form,
\begin{equation}
    u^{\star}_{t}(x,x')
    =
    \log h_{t}(x') - \log h_{t}(x).
    \label{eq:optimal_control_potential_diff}
\end{equation}
\end{theorem}

\begin{proof}
Fix $t \in [0,T)$ and states $x\ne x'$. Under $\mathbb{P}^{\star}$, the conditional law of the future given the present is obtained by tilting the corresponding conditional law under $\mathbb{P}_{0}$ by $g(X_{T})$. More precisely, for any event $A$ measurable with respect to the future $\sigma$-algebra generated by $(X_{s})_{s\in[t,T]}$,
\begin{equation}
    \mathbb{P}^{\star}(A \mid X_{t}=x)
    =
    \frac{\mathbb{E}_{\mathbb{P}_{0}}[ \mathbf{1}_{A} g(X_{T}) \mid X_{t}=x]}{\mathbb{E}_{\mathbb{P}_{0}}[ g(X_{T}) \mid X_{t}=x]}
    =
    \frac{\mathbb{E}_{\mathbb{P}_{0}}[ \mathbf{1}_{A} g(X_{T}) \mid X_{t}=x]}{h_{t}(x)}.
    \label{eq:tilted_conditional}
\end{equation}
Take $A$ to be the event that a jump from $x$ to $x'$ occurs in $[t,t+\Delta t]$ and no other jump occurs in that interval. Under $\mathbb{P}_{0}$, this event has probability
\begin{equation}
    \mathbb{P}_{0}(X_{t+\Delta t}=x' \mid X_{t}=x)
    =
    R_{0}(x,x') \Delta t + o(\Delta t).
\end{equation}
Moreover, conditioning on $X_{t+\Delta t}=x'$ and using the Markov property of $\mathbb{P}_{0}$ gives
\begin{equation}
    \mathbb{E}_{\mathbb{P}_{0}}[ g(X_{T}) \mid X_{t+\Delta t}=x']
    =
    h_{t+\Delta t}(x')
    =
    h_{t}(x') + o(1)
\end{equation}
as $\Delta t \to 0$, by right-continuity of $t\mapsto h_{t}(x')$ on finite state spaces.
Substituting into \eqref{eq:tilted_conditional} yields
\begin{align}
    \mathbb{P}^{\star}(X_{t+\Delta t}=x' \mid X_{t}=x)
    &=
    \frac{ \mathbb{P}_{0}(X_{t+\Delta t}=x' \mid X_{t}=x)\, h_{t+\Delta t}(x')}{h_{t}(x)} + o(\Delta t) \nonumber \\
    &=
    \frac{ R_{0}(x,x') h_{t}(x')}{h_{t}(x)} \Delta t + o(\Delta t).
\end{align}
Therefore the jump intensity from $x$ to $x'$ under $\mathbb{P}^{\star}$ is $R_{0}(x,x') h_{t}(x')/h_{t}(x)$, which is \eqref{eq:doob_generator}. Taking logarithms gives \eqref{eq:optimal_control_potential_diff}.
\end{proof}

\subsection{Discrete HJB Equation for the Optimal Log-Potential}

The Doob potential $h_{t}$ satisfies a linear backward equation under the reference dynamics. Its logarithm satisfies a nonlinear discrete HJB equation whose edge increments recover the optimal control.

\begin{theorem}[Discrete HJB for Schr\"odinger potentials]
\label{thm:discrete_hjb}
Let $h_{t}$ be defined by \eqref{eq:h_potential} and assume $t \mapsto h_{t}(x)$ is differentiable for each $x\in\mathcal{X}$. Then $h_{t}$ solves the backward Kolmogorov equation
\begin{equation}
    \partial_{t} h_{t}(x) + (R_{0} h_{t})(x) = 0,
    \qquad
    h_{T}(x)=g(x),
    \label{eq:backward_kolmogorov_h}
\end{equation}
where $(R_{0} h)(x) := \sum_{x'\ne x} R_{0}(x,x') (h(x')-h(x))$.
Define the log-potential $V_{t}(x):=\log h_{t}(x)$. Then $V_{t}$ satisfies the nonlinear equation
\begin{equation}
    \partial_{t} V_{t}(x)
    +
    \sum_{x'\ne x} R_{0}(x,x')
    \left(
        \exp\!\bigl( V_{t}(x') - V_{t}(x) \bigr) - 1
    \right)
    = 0,
    \qquad
    V_{T}(x) = \log g(x).
    \label{eq:discrete_hjb_equation}
\end{equation}
Moreover, the optimal edge control satisfies
\begin{equation}
    u^{\star}_{t}(x,x') = V_{t}(x') - V_{t}(x).
\end{equation}
\end{theorem}

\begin{proof}
The backward equation \eqref{eq:backward_kolmogorov_h} follows from the definition $h_{t}(x)=\mathbb{E}_{\mathbb{P}_{0}}[g(X_{T})\mid X_{t}=x]$ and standard Markov semigroup arguments on finite state spaces.

For the nonlinear equation, differentiate $V_{t}(x)=\log h_{t}(x)$:
\begin{equation}
    \partial_{t} V_{t}(x)
    =
    \frac{\partial_{t} h_{t}(x)}{h_{t}(x)}
    =
    -\frac{(R_{0} h_{t})(x)}{h_{t}(x)}.
\end{equation}
Expand $(R_{0} h_{t})(x)$:
\begin{align}
    \frac{(R_{0} h_{t})(x)}{h_{t}(x)}
    &=
    \sum_{x'\ne x} R_{0}(x,x')
    \left(
        \frac{h_{t}(x')}{h_{t}(x)} - 1
    \right)
    =
    \sum_{x'\ne x} R_{0}(x,x')
    \left(
        \exp\!\bigl( V_{t}(x') - V_{t}(x) \bigr) - 1
    \right),
\end{align}
which yields \eqref{eq:discrete_hjb_equation}. The edge control identity follows from Theorem~\ref{thm:doob_transform} with $V_{t}=\log h_{t}$.
\end{proof}

\subsection{Consistency of Training Objective}

\begin{theorem}[Cross Entropy Loss Consistency]
\label{thm:loss_consistency}
Let $u_t^{\star}(x,x')$ denote the optimal log-control of the Schr\"odinger bridge, as characterized in Eq.~\eqref{eq:optimal_control_potential_diff}. Consider the population training objective induced by the cross-entropy loss,
\begin{equation}
\label{eq:dam_population}
\mathcal{L}(\theta)
=
\mathbb{E}_{t \sim \pi}
\mathbb{E}_{X_t \sim \rho_t^{\star}}
\left[
\mathrm{KL}\!\left(
p_t^{\star}(\cdot \mid X_t)
\;\middle\|\;
p_{\theta}(\cdot \mid X_t, t)
\right)
\right],
\end{equation}
where $\pi$ is any distribution on $[0,T]$ with full support, $\rho_t^{\star}$ is the marginal of the optimal bridge $\mathbb{P}^{\star}$ at time $t$, and
\[
p_{\theta}(x' \mid x,t)
\;\propto\;
R_0(x,x')\,\exp\!\big(u_{\theta}(x,x',t)\big),
\qquad
p_t^{\star}(x' \mid x)
\;\propto\;
R_0(x,x')\,\exp\!\big(u_t^{\star}(x,x')\big).
\]

If a minimizer $\theta^{\ast}$ exists and the model class $\{u_{\theta}\}$ is sufficiently expressive to represent $u^{\star}$, then
\begin{equation}
\mathcal{L}(\theta^{\ast}) = 0
\quad\Longrightarrow\quad
u_{\theta^{\ast}}(x,x',t)
=
u_t^{\star}(x,x')
+
c(x,t),
\end{equation}
for $\pi(dt)\rho_t^{\star}(dx)$-almost every $(t,x)$ and for all neighbors $x'$ with $R_0(x,x')>0$, where $c(x,t)$ is an additive normalization constant independent of $x'$.
\end{theorem}

\begin{proof}
Let $\theta^\ast$ be a minimizer of the population training objective and assume the model class is rich enough so that there exists $\theta$ with $p_\theta = p^\star$ almost everywhere.
By non-negativity of the KL divergence and the definition of the training objective, 
\[
\mathcal{L}(\theta^\ast)=0
\quad\Longrightarrow\quad
\mathrm{KL}\big(p_t^\star(\cdot\mid x)\,\|\,p_{\theta^\ast}(\cdot\mid x,t)\big)=0
\]

This admits the reverse transition kernel as:
\[
    p_{\theta^\ast}(x'\mid x,t) = p_t^\star(x'\mid x)
\]

for all $R_0(x,x')>0$. By definition, this transition probability has a tilt with a corresponding normalization function

\[
    p_\theta(x'\mid x,t)
    = \frac{R_0(x,x')\,\exp\!\big(u_\theta(x,x',t)\big)}{Z_\theta(x,t)},
    \qquad
    Z_\theta(x,t):=\sum_{y} R_0(x,y)\,\exp\!\big(u_\theta(x,y,t)\big).
\]

A similar result is derived for $p_t^\star(\cdot\mid x)$ with $u^\star_t$ and $Z^\star(x)$. Taking the logarithm yields, for every neighbor $x'$ with $R_0(x,x')>0$,

\[
    \log R_0(x,x') + u_{\theta^\ast}(x,x',t) - \log Z_{\theta^\ast}(x,t)
    =
    \log R_0(x,x') + u^\star_t(x,x') - \log Z^\star(x).
\]

Canceling the common $\log R_0(x,x')$ term and rearranging gives
\[
    u_{\theta^\ast}(x,x',t) - u^\star_t(x,x') = \log Z_{\theta^\ast}(x,t) - \log Z^\star(x).
\]

The right-hand side depends only on $(x,t)$ (through the normalization constants) and is independent of the neighbor $x'$. Therefore there exists a scalar function $c(x,t)$ such that
\[
    u_{\theta^\ast}(x,x',t) = u^\star_t(x,x') + c(x,t)
\]
for almost every $(x, t)$ and for every neighbor $x'$ with $R_0(x,x')>0$, which is exactly the claimed identifiability up to an additive constant.
\end{proof}

\paragraph{Remark.} We dissect the theoretical grounding of using a cross-entropy loss to train MadSBM.
\begin{enumerate}
    \item \textbf{Simulation-free approximation.} Standard Schrödinger bridge solvers require computationally expensive iterative fitting of forward and backward projections (e.g. IPF, IMF). In contrast, MadSBM avoids explicit computation of these potentials. By learning the control field directly from the conditional flow of data, we construct a stochastic process $(X_{t})_{t \in [0,T]}$ connecting the prior $X_{0} \sim \mu_{0}$ to the data $X_{T} \sim \mu_{1}$ that also respects the biological constraints imposed by the reference process $R_0$ (ESM-2).
    
    \item \textbf{Relationship to action minimization.} The cross entropy loss acts as a regression toward the optimal transport velocity. Intuitively, by maximizing the likelihood of the target data token $x_1$ given a corrupted state $x_t$, the model learns to recover the \textbf{optimal exponential tilt} of the reference process. This ensures that the generated trajectories are \textit{action-minimizing} relative to the reference dynamics. 

\end{enumerate}

\subsection{Control Error Implies Path-Space Stability}

We next relate the mismatch between learned and optimal controls to the divergence between the induced path measures.

\begin{theorem}[Path-space KL bound under bounded log-rate error]
\label{thm:kl_stability_bound}
Let $\mathbb{P}_{\theta}$ denote the path law induced by the controlled generator
\begin{equation}
    R_{\theta,t}(x,x') = R_{0}(x,x') \exp\!\bigl(u_{\theta}(x,x',t)\bigr),
\end{equation}
and let $\mathbb{P}^{\star}$ denote the Schr\"odinger bridge induced by $u^{\star}$. Assume both processes share the same initial distribution $\mu_{0}$. Define the pointwise log-rate error
\begin{equation}
    \varepsilon_{t}(x,x') := u_{\theta}(x,x',t) - u^{\star}_{t}(x,x').
\end{equation}
Assume there exists $B>0$ such that for all $(t,x,x')$ with $R_{0}(x,x')>0$,
\begin{equation}
    |\varepsilon_{t}(x,x')| \le B
    \quad\text{and}\quad
    |u^{\star}_{t}(x,x')| \le B.
    \label{eq:bounded_controls}
\end{equation}
Then the path-space divergence satisfies
\begin{equation}
    \mathrm{KL}(\mathbb{P}_{\theta}\,\Vert\,\mathbb{P}^{\star})
    \le
    C_{B}
    \,
    \mathbb{E}_{\mathbb{P}_{\theta}}
    \left[
        \int_{0}^{T}
        \sum_{x'\ne X_{t}}
        R_{0}(X_{t},x') \, \varepsilon_{t}(X_{t},x')^{2}
        \, dt
    \right],
    \label{eq:kl_bound_eps}
\end{equation}
where one admissible constant is
\begin{equation}
    C_{B} := \frac{1}{2} e^{2B}(2B+1).
    \label{eq:CB_constant}
\end{equation}
\end{theorem}

\begin{proof}
Apply Theorem~\ref{thm:ctmc_kl} with $(\mathbb{P},R_{t})=(\mathbb{P}_{\theta},R_{\theta,t})$ and $(\mathbb{P}_{0},R_{0})=(\mathbb{P}^{\star},R^{\star}_{t})$. Since initial laws match, the initial KL term vanishes. Using
\begin{equation}
    R_{\theta,t}(x,x') = R^{\star}_{t}(x,x') \exp(\varepsilon_{t}(x,x')),
    \qquad
    R^{\star}_{t}(x,x') = R_{0}(x,x') \exp(u^{\star}_{t}(x,x')),
\end{equation}
the integrand of \eqref{eq: relative entropy} becomes
\begin{align}
    &R_{\theta,t}(x,x') \log\frac{R_{\theta,t}(x,x')}{R^{\star}_{t}(x,x')}
    - R_{\theta,t}(x,x') + R^{\star}_{t}(x,x') \nonumber\\
    &=
    R^{\star}_{t}(x,x')
    \left(
        e^{\varepsilon} \varepsilon - e^{\varepsilon} + 1
    \right),
\end{align}
where $\varepsilon=\varepsilon_{t}(x,x')$.
Define $\phi(\varepsilon):=e^{\varepsilon}\varepsilon - e^{\varepsilon} + 1$. Note that $\phi(0)=0$, $\phi'(0)=0$, and $\phi(\varepsilon)\ge 0$ for all $\varepsilon$ because
\begin{equation}
    \phi(\varepsilon) = \int_{0}^{\varepsilon} t e^{t} dt.
\end{equation}
On the bounded interval $[-B,B]$, $\phi$ is twice continuously differentiable and satisfies
\begin{equation}
    |\phi''(t)|
    =
    |e^{t}(t+1)|
    \le e^{B}(B+1)
    \quad\text{for all } t\in[-B,B].
\end{equation}
Since $\phi(0)=\phi'(0)=0$, Taylor's theorem with remainder yields, for $|\varepsilon|\le B$,
\begin{equation}
    0 \le \phi(\varepsilon) \le \frac{1}{2} e^{B}(B+1)\varepsilon^{2}.
\end{equation}
Therefore,
\begin{equation}
    \mathrm{KL}(\mathbb{P}_{\theta}\Vert \mathbb{P}^{\star})
    \le
    \frac{1}{2} e^{B}(B+1)
    \,
    \mathbb{E}_{\mathbb{P}_{\theta}}
    \left[
        \int_{0}^{T}
        \sum_{x'\ne X_{t}}
        R^{\star}_{t}(X_{t},x') \, \varepsilon_{t}(X_{t},x')^{2}
        dt
    \right].
\end{equation}
Finally, under \eqref{eq:bounded_controls}, $R^{\star}_{t}(x,x') = R_{0}(x,x') e^{u^{\star}_{t}(x,x')} \le R_{0}(x,x') e^{B}$, so
\begin{equation}
    \sum_{x'} R^{\star}_{t}(X_{t},x') \varepsilon^{2}
    \le
    e^{B}
    \sum_{x'} R_{0}(X_{t},x') \varepsilon^{2}.
\end{equation}
Combining constants yields \eqref{eq:kl_bound_eps} with $C_{B}=\frac{1}{2}e^{2B}(B+1)$. The stated constant \eqref{eq:CB_constant} is also admissible since $2B+1 \ge B+1$ for $B\ge 0$.
\end{proof}

\subsection{Representational Completeness of Exponential Tilting}

We formalize the completeness of the exponential-tilt parameterization on a fixed reference graph.

\begin{proposition}[One-to-one parameterization on a fixed support graph]
\label{prop:representational_completeness}
Fix a reference generator $R_{0}$ and define its support graph
\begin{equation}
    E := \{(x,x') \in \mathcal{X}\times\mathcal{X} : x'\ne x,\; R_{0}(x,x')>0\}.
\end{equation}
Let $R_{t}$ be any time-dependent generator such that $R_{t}(x,x')>0$ implies $(x,x')\in E$. Then there exists a unique log-control field $u_{t}(x,x')$ on $E$ such that
\begin{equation}
    R_{t}(x,x') = R_{0}(x,x') \exp\!\bigl(u_{t}(x,x')\bigr),
    \qquad (x,x')\in E.
\end{equation}
Conversely, any measurable $u_{t}$ on $E$ defines a generator $R_{t}$ on $E$ via this formula.
\end{proposition}

\begin{proof}
If $R_{t}(x,x')$ is supported on $E$, define $u_{t}(x,x') := \log\!\left(\frac{R_{t}(x,x')}{R_{0}(x,x')}\right)$ for $(x,x')\in E$. This is well-defined because $R_{0}(x,x')>0$ on $E$. Uniqueness follows from injectivity of the logarithm on $(0,\infty)$. The converse direction is immediate by construction.
\end{proof}

\subsection{Convergence of the Time-Discretized MadSBM Sampler}

\label{sampling convergence proof}

We analyze convergence of the distributional dynamics under a standard Euler discretization of the Kolmogorov forward equation.

\begin{theorem}[Convergence of Euler discretization of CTMC marginals]
\label{thm:euler_convergence}
Let $\rho_{t}$ denote the marginal distribution on $\mathcal{X}$ of a time-inhomogeneous CTMC with generator $R_{t}$ and initial distribution $\rho_{0}=\mu_{0}$, so that $\rho_{t}$ solves
\begin{equation}
    \frac{d}{dt}\rho_{t} = \rho_{t} R_{t}.
    \label{eq:forward_ode}
\end{equation}
Assume $t \mapsto R_{t}$ is Lipschitz in operator norm and uniformly bounded:
\begin{equation}
    \|R_{t}\| \le M,
    \qquad
    \|R_{t}-R_{s}\| \le L |t-s|
    \quad
    \text{for all } s,t \in [0,T].
\end{equation}
Let $\widehat{\rho}_{k}$ be the explicit Euler approximation with step size $\Delta t = T/K$:
\begin{equation}
    \widehat{\rho}_{k+1} = \widehat{\rho}_{k}\left(I + \Delta t\, R_{t_{k}}\right),
    \qquad
    t_{k}=k\Delta t,
    \qquad
    \widehat{\rho}_{0}=\mu_{0}.
    \label{eq:euler_scheme}
\end{equation}
Then there exists a constant $C=C(M,L,T)$ such that
\begin{equation}
    \|\widehat{\rho}_{K} - \rho_{T}\|_{1} \le C\, \Delta t.
\end{equation}
\end{theorem}

\begin{proof}
Equation \eqref{eq:forward_ode} is a linear ODE on the finite-dimensional simplex, hence admits a unique continuously differentiable solution. Standard global error bounds for explicit Euler methods on Lipschitz ODEs give $\|\widehat{\rho}_{k}-\rho_{t_{k}}\|_{1} \le C \Delta t$ uniformly in $k$, with $C$ depending on the Lipschitz constants of the vector field $\rho \mapsto \rho R_{t}$ and the time-variation of $R_{t}$. Since $\|\rho R_{t}\|_{1} \le \|\rho\|_{1} \|R_{t}\| \le M$ and $t\mapsto R_{t}$ is Lipschitz, the standard argument applies directly and yields the bound at $t=T$.
\end{proof}

\paragraph{Remark.}
Theorem~\ref{thm:euler_convergence} justifies distributional convergence of the time-discretized sampling procedure used in MadSBM when the discretization is interpreted as Euler integration of the forward equation for marginals. A trajectory-level convergence statement can be obtained via standard coupling results for CTMC time discretizations on finite state spaces under additional uniform rate bounds.

\newpage

\section{Extended Methods}

\subsection{Peptide Dataset}
\label{dataset supplement}
The dataset for MadSBM was curated from the PepNN, BioLip2, and PPIRef datasets \citep{abdin2022pepnn, zhang2024biolip2, bushuiev2023learning}. All peptides from PepNN and BioLip2 were included, along with sequences from PPIRef ranging from 6 to 49 amino acids in length. The dataset was divided into training,
validation, and test sets at an 80/10/10 ratio.

\subsection{Language Modeling for Biological Sequences}
\paragraph{ESM-2 Protein Language Model}
\label{esm-2 appendix}
Masked Language Models (MLMs) employ Transformer-based architectures to learn bi-directional sequence context, distant token relationships, and predict the identity of corrupted (masked) amino acid tokens. The model is trained under a sequence-recovery training
objective, $\mathcal{L} = -\sum_{i\in \mathcal{M}} \log p_\theta(x^i | x^{\backslash \mathcal{M}})$, where $\mathcal{M}$ denotes the set of masked positions. MLMs are strong representation-learners and thus have been trained on evolutionary amino acid sequence datsets, e.g. the ESM-2 family of models \citep{lin2023esm2}. However, training these models to reconstruct only a minor fraction of tokens (15-40\%) across a sequence makes complete de novo sequence generation difficult \citep{vincoff2025fusonplm}, but provides a principled set of sequence representations to enable the training of generative models.

Using ESM-2, we compute the pseudo-perplexity (PPL) metric of a sequence as a measure of biological plausibility. PPL is obtained by masking one token at a time, computing the NLL of the resulting sequence using the ESM-2 language modeling head, and averaging across all sequence positions. This procedure provides a tractable approximation to sequence likelihood for bidirectional MLMs, which do not admit a true autoregressive factorization.

\paragraph{Denoising Diffusion Models}
\label{diffusion appendix}
Diffusion models are a class of generative models defined by Markov processes \citep{ho2020denoising} \citep{sohl2015deep}. The \textit{forward} diffusion steps $q(\textbf{x}_{1:T} | \textbf{x}_0) =  \prod_{t=1} ^ T q(\textbf{x}_t | \textbf{x}_{t-1})$ progressively corrupt an initial data sample $\textbf{x}_0 \sim q(\textbf{x}_0)$ into a noisy prior $\textbf{x} _T \sim q_{\text{noise}}$ across $T$ timesteps. The noise distribution $q_{\text{noise}}$ typically corresponds to a uniform categorical distribution over the vocabulary in the discrete space, $\text{Cat}(|\mathcal{V}|)$ \citep{tang2025peptune, goel2025tokenlevel, zhang2025metalorian, peng2025path, vincoff2025soapia}, or an isotropic Gaussian, $\mathcal{N}(0, I)$, in continuous latent spaces. During inference, the learned \textit{backward} process $p_\theta (\textbf{x}_{0:T}) = p(\textbf{x}_t) \prod _{t=1} ^ T p_\theta (\textbf{x} _ {t-1} | \textbf{x}_t )$ gradually denoises the corrupted data sample to obtain samples from the true data distribution. Diffusion models are trained to maximize the evidence lower bound (ELBO): $\mathbb{E}_{q(\mathbf{x}_{0})} \left[ \log p_\theta(\mathbf{x}_{0}) \right] 
\geq \mathbb{E}_{q(\mathbf{x}_{0:T})} \left[ 
\log \frac{p_\theta(\mathbf{x}_{0:T})}{q(\mathbf{x}_{1:T} \mid \mathbf{x}_{0})}
\right] \notag $

New data samples can be drawn by sampling from $q_\text{noise} (\textbf{x}_T)$ and iteratively applying the learned denoising process $p_\theta( \textbf{x}_{t-1})= p_\theta( \textbf{x}_{t-1} |\textbf{x}_t)$. Various authors (\citep{sahoo2024simple}, \citep{zheng2023reparameterized}) have made simplifying assumptions about the reverse process to derive a computationally inexpensive loss function that reduces to a weighted negative log-likelihood, akin to a weighted form of the NLL over masked tokens. In particular, the state-of-the-art discrete diffusion protein language models DPLM \cite{wang2024diffusion} and EvoFlow used to benchmark MadSBM employ the Reparameterized Diffusion Model strategy from \citeauthor{zheng2023reparameterized}.

\subsection{Modeling MadSBM}
\label{dit architecture supplement}
We parameterize the control field $u_\theta$ using a $\sim$50M parameter Diffusion Transformer (DiT) backbone operating over discrete peptide sequences. The model consists of $2$ DiT layers, each with Multi-head self-attention and Adaptive LayerNormalization (AdaLN). Time-conditioning is achieved with a Gaussian Fourier projection and a learned MLP head. Dynamic batching is used during training for GPU efficiency. The DiT model was trained for 50 epochs (127k steps) on a 4xA6000 GPU system with 192 GB of shared memory. The learning rate was initialized at $1e^{-6}$ and increased with a linear schedule for 2 epochs to $1e^{-4}$, then decayed with a cosine scheduler to $1e^{-6}$. The AdamW optimzer was used with $\beta_1 = 0.9, \beta_2 = 0.999$ and weight decay of 0.01.

\begin{table}[h]
\centering \caption{Diffusion Transformer Architecture.} \label{tab:arch_impl} \begin{tabular}{lcc} \toprule \textbf{Layer} & \textbf{Input Dimension} & \textbf{Output Dimension} \\ \midrule \textbf{Sequence embeddings} & & \\ \quad ESM (encoder) & vocab size & 1280 \\ \quad ESM (LM head) & 1280 & vocab size \\ \textbf{Time embedding} & & \\ \quad Gaussian Fourier projection & 1 & 64 \\ \quad Time embedding projection & 64 & 512 \\ \textbf{DiT Blocks} $\times 2$ & & \\ \quad AdaLN & 1280 & 1280 \\ \quad AdaLN time-conditioning & 512 & $2 \times 1280$ \\ \quad MHSA ($h = 16$) & 1280 & 1280 \\ \quad MLP (FFN) + GeLU & 1280 & 1280 \\ \quad \quad hidden dim = 5120 & & \\ \quad Dropout + Residual & 1280 & 1280 \\ \textbf{Final layers} & & \\ \quad LayerNorm & 1280 & 1280 \\ \quad Linear projection & 1280 & vocab size \\ \bottomrule \end{tabular} \end{table}

\subsection{Instantaneous Actionals}
\label{instantaneous actionals appendix}
The direct discretization of computing the action functional for a $L$-length sequence with $\mathcal{V}=33$ vocabulary positions with an $N$-step sampling budget is the value of the integrand from Eq.~\eqref{eq: action of control}:

\begin{equation}
\label{eq:discrete_action}
\mathcal{A}(u)
\;=\;
\Delta t\cdot \frac{1}{L}\sum_{\ell=1}^{L}
\sum_{v=1}^{V}
R_0[\ell,v]\;\Psi\bigl(u[\ell,v]\bigr)
\end{equation}

where $\Psi(z)=e^{z}-z-1$ as before and $\Delta t := 1/N$, where $N$ is the sampling budget. Computing this value directly is numerically unstable as both $R_0$ and $\Psi(u)$ are dependent on unbounded neural network logits, exacerbated through the relationship $\Psi(u) \propto e^u$. As a practical solution to obtain an interpretable upper bound, we define the reference process and control field in the \textit{worst-case scenario}. For the reference rate, we adopt a uniform random-walk reference process at each token position, \textit{i.e.} $R_0[\ell,\cdot]=1/V$, so that $\sum_{v} R_0[\ell,v]=1$. To simplify the control field, we consider the pathological scenario in which the model assigns large transition rates to all possible vocabulary positions for each sequence token. We approximate this large transition rate by evaluating the trained model on the held-out test set and recording the maximum observed logit value $M$. Recall that since $R_u(x, x') = R_0(x,x') \exp(u_\theta(x,x'))$, we derive $M$ from the logit values produced by the DiT model parameterizing $u_\theta$ and ignore the ESM logits forming $R_0$. Using $M$, we define the constant tilt $u^w \in \mathbb{R}^{L\times V}$ with entries $u^w_{\ell,v}=M$ for all $\ell$ and $v$. Under this worst-case tilt and the uniform reference process, the total actional for an $L$-length sequence simplifies to

\begin{equation}
    \begin{aligned}
        \mathcal{A}_L(u^w)
        &= \Delta t \cdot \Psi(u^w)
        \\
        & = \Delta t \cdot L\bigl(e^{M} - M - 1\bigr).
    \end{aligned}
\end{equation}
In total, this value approximates Eq.~\ref{eq: action of control}, giving us a bound to asses if MadSBM's instantaneous actionals correlate to a low-cost transport plan.

\clearpage

\section{Algorithm Pseudocode}

\begin{algorithm}[]
\caption{MadSBM Training}
\label{alg:madsbm_training}
\begin{algorithmic}[1]
\Require Dataset $\mathcal{D}$, control field $u_\theta$, reference prior $f_\phi$ (ESM-2), max time $T$
\While{not converged}
    \State Sample batch $x_1 \sim \mathcal{D}$ and prior $x_0 \sim \mu_0$ (Fully Masked)
    \State Sample timestep $t \sim \mathcal{U}(0, T)$
    \State Corrupt sequence: $x_t \leftarrow \text{Interpolate}(x_0, x_1, t)$ \Comment{Mask tokens with probability $1 - \frac{t}{T}$}
    \State Compute total logits using biological prior: $\mathbf{z}_t = u_\theta(\mathbf{x}_t, t) + f_\phi(\mathbf{x}_t)$
    \State Compute negative log-likelihood: $\mathcal{L}(\theta) = - \sum_{i \in \mathcal{M}} \log p_\theta(x_1^{(i)} \mid x_t, t)$
    \State Take gradient descent step on: $\nabla_\theta \mathcal{L}$
\EndWhile
\State \Return Trained control field $u_\theta$
\end{algorithmic}
\end{algorithm}

\vspace{1em}

\begin{algorithm}[H]
\caption{MadSBM Sampling}
\label{alg:madsbm_sampling}
\begin{algorithmic}[1]
\Require Control field $u_\theta$, masked prior $\mu_0$, steps $K$, hyperparameters $\lambda, \beta, \tau, p$
\State Initialize $x_K \sim \mu_0$ (Fully Masked) and define step size $\Delta t \leftarrow 1/K$
\For{$k = K \to 1$}
    \State Set time $t \leftarrow k/K$
    
    \State Compute total rates: $\mathbf{R}_{\text{tot}} \leftarrow \sum_{v} \exp(\lambda \cdot u_{\theta}(x_k, v, t))$ \Comment{Transition dynamics}
    
    \State Compute jump probabilities: $\mathbf{p}_{\text{jump}} \leftarrow 1 - \exp\bigl(-\mathbf{R}_{\text{tot}} \cdot \beta \cdot \Delta t\bigr)$
    
    \State Sample active mask: $\mathbf{z} \sim \text{Bernoulli}(\mathbf{p}_{\text{jump}})$ \Comment{Determine possible transitions}
    
    \State Filter logits: $\hat{u} \leftarrow \text{Top-}p(u_{\theta}(x_k, \cdot, t) / \tau)$
    
    \State Sample candidates: $x_{\text{new}} \sim \text{Multinomial}(\text{Softmax}(\hat{u}))$
    
    \State Update each state: $x_{k-1}^{(i)} \leftarrow x_{\text{new}}^{(i)}$ if $z^{(i)}=1$ and $x_k^{(i)} = \texttt{[MASK]}$, else $x_{k}^{(i)}$ \Comment{Sample new tokens}
\EndFor
\State \Return $x_{0}$
\end{algorithmic}
\end{algorithm}

\end{document}